\documentclass{article} %
\usepackage{iclr2027_conference,times}

\usepackage{graphicx}
\usepackage{amsmath,amssymb,amsthm}
\usepackage{mathtools}
\usepackage{booktabs}
\usepackage{array}
\usepackage{longtable}
\usepackage{multirow}
\usepackage{placeins}
\usepackage{caption}
\usepackage{subcaption}
\usepackage{url}
\usepackage{hyperref}
\hypersetup{hidelinks,
  pdftitle={mAVE: A Watermark for Joint Audio-Visual Generation Models},
  pdfauthor={Luyang Si, Leyi Pan, Dongsheng Ma, Lijie Wen}}
\usepackage[capitalise,noabbrev]{cleveref}
\theoremstyle{plain}
\newtheorem{theorem}{Theorem}
\newtheorem{lemma}{Lemma}
\theoremstyle{definition}

\usepackage{xspace}
\makeatletter
\DeclareRobustCommand\onedot{\futurelet\@let@token\@onedot}
\def\@onedot{\ifx\@let@token.\else.\null\fi\xspace}
   
\def\ie{\emph{i.e}\onedot}

\makeatother

\crefname{theorem}{Theorem}{Theorems}
\crefname{lemma}{Lemma}{Lemmas}
\crefname{equation}{Eq.}{Eqs.}

\title{mAVE: A Watermark for Joint Audio-Visual\\ Generation Models}

\author{Luyang Si$^{1}$, Leyi Pan$^{1}$, Dongsheng Ma$^{2}$, Lijie Wen$^{1}$ \\
$^{1}$Tsinghua University \qquad $^{2}$Peking University
}

\iclrfinalcopy
\begin{document}

\maketitle
\pagestyle{plain}
\suppressfloats[t]

\begin{abstract}
Watermarking joint audio-visual generation supports vendor copyright protection and content provenance. However, independently valid audio and video watermarks do not establish a shared generation session. An adversary can splice watermarked modalities from different sessions, causing the pair to be mistaken for the vendor's original joint output. We introduce mAVE (Manifold Audio-Visual Entanglement), a training-free watermarking framework that strengthens vendor attribution through session binding in native joint audio-visual diffusion transformers. mAVE separates public record retrieval from secret session authentication: a fixed public index locates the server record, while a randomized payload binds audio bits to a session-keyed video grid through a cryptographic digest. One prompt-conditioned joint inversion supports provider-assisted verification of both modalities against a session record, without modifying generator weights or training auxiliary watermark networks. Our analysis establishes implementation-matched distribution preservation and a full-initialization routing/clipping budget, alongside adaptive session-pool security and stable local-perturbation bounds. Experiments on LTX-2 and MOVA show comparable generation quality. mAVE achieves 99.8\% true-positive rate and 0\% observed false-positive rate in the evaluated swap test, and retains 99.2\% true-positive rate under FrameAvg temporal averaging. Same-prompt and similarity-selected swaps further test session authentication beyond perceptual compatibility.
\end{abstract}

\section{Introduction}

Generative AI is making realistic audiovisual content easier to produce and distribute. Commercial video-generation platforms~\citep{openai2025sora2,deepmind2025veo3,kling2026kling3,bytedance2026seedance2} and openly available models~\citep{hacohen2026ltx2,openmossteam2026mova} expand creative applications while increasing opportunities to repurpose and misattribute synthetic media. Watermarking provides a proactive signal for synthetic-content identification and source attribution, supporting vendor copyright protection beyond judgments based on visual or auditory realism~\citep{fernandez2024videoseal,sanroman2024proactive}. As generation moves from individual modalities to synchronized audio-video outputs, that protection must also account for how the modalities were produced together.

Native joint generators denoise audio and video within one inference request, producing a paired output with a shared generation history. Yet valid video and audio watermarks need not certify that history. An adversary can combine watermarked video from one session with watermarked audio from another, causing the pair to be mistaken for the vendor's original joint output. The individual marks survive while the attributed pairing is false (\cref{fig:teaser}).

\begin{figure}[t]
\centering
\includegraphics[width=0.9\linewidth]{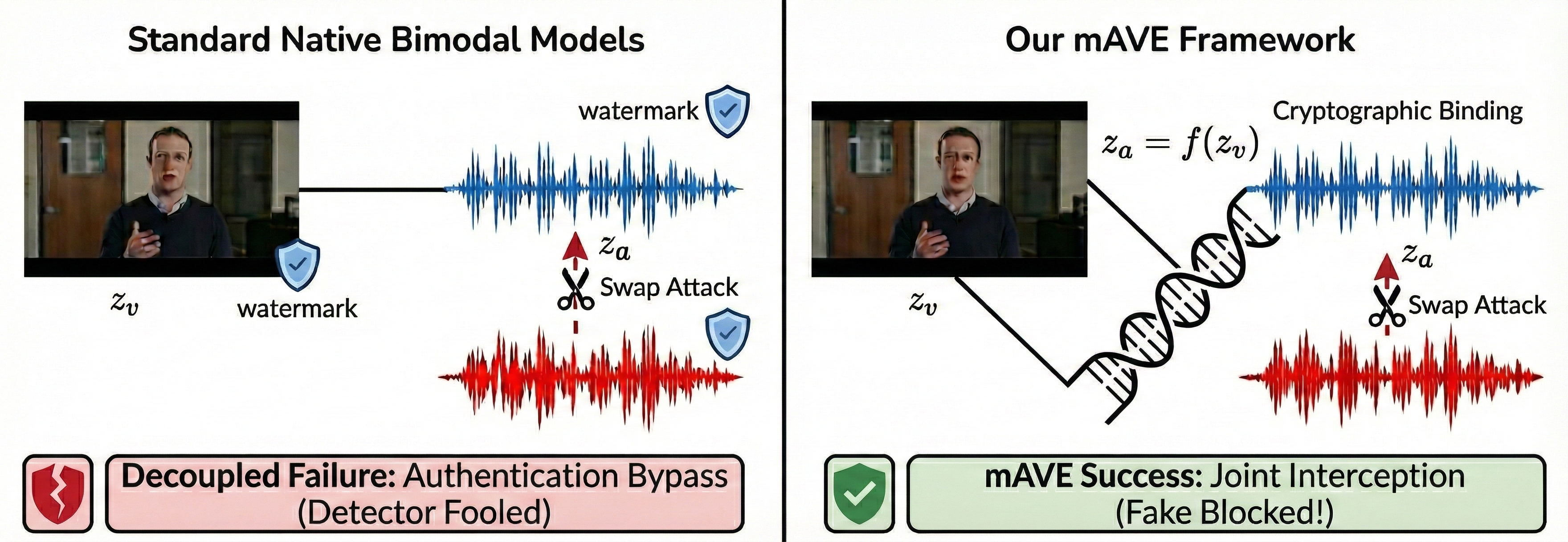}
\caption{\textbf{Session binding with mAVE.} Left: independently valid watermarks can survive cross-session AV splicing. Right: mAVE binds the initial latent pair, enabling the verifier to test shared session provenance.}
\label{fig:teaser}
\end{figure}

Let $Video_{wm}$ and $Audio_{wm}$ denote valid modality-level detections. Neither accepting either mark ($Video_{wm}\vee Audio_{wm}$) nor requiring both ($Video_{wm}\wedge Audio_{wm}$) resolves this gap. Perceptually plausible swaps can also pass synchronization checks~\citep{chung2016syncnet}. Reliable attribution therefore requires protection at the level of the \emph{generation-session pair}, beyond two separately valid modalities.

We introduce \textbf{mAVE} (Manifold Audio-Visual Entanglement), which places this session constraint in a joint generator's initial latents. A video-grid digest determines audio binding bits; inverse-transform sampling embeds the randomized payload, and one prompt-conditioned joint inversion recovers both modalities. A public index locates the server record for verification. mAVE requires neither generator fine-tuning nor an auxiliary watermark network, enabling integration with existing joint generators through initialization and verification.

On LTX-2 and MOVA, experiments show comparable generation quality. In the evaluated swap test, mAVE achieves 99.8\% true-positive rate (TPR) and 0\% observed false-positive rate (FPR). mAVE also performs well under temporal attacks: FrameAvg retains 99.2\% TPR with 0\% observed FPR. Authentic-pair TPR remains 99.0\% after MP4 compression and 98.1\% after diffusion-based regeneration with audio reconstruction.

Our contributions are threefold:
\begin{itemize}
\item \textbf{Threat formulation:} We formalize cross-session AV swaps and distinguish modality-level watermark detection from same-session pair authentication.
\item \textbf{Method and theory:} We introduce training-free session binding without changing generator weights, combining public record retrieval and one joint inversion with implementation-aware sampling guarantees, an explicit full-prior distortion budget, and adaptive session-pool security bounds.
\item \textbf{Evaluation:} Quality remains comparable on LTX-2 and MOVA. The evaluated swap test yields 99.8\% TPR and 0\% observed FPR; FrameAvg, MP4 compression, and regeneration retain 99.2\%, 99.0\%, and 98.1\% TPR, respectively. Matched server-reference controls assess the binding construction.
\end{itemize}

\section{Related Work}

\textbf{Joint AV generation.} LTX-2~\citep{hacohen2026ltx2} and MOVA~\citep{openmossteam2026mova} jointly denoise audio and video, exchanging information during generation. Their coupled initialization offers a common point for embedding a session constraint; flow-based sampling~\citep{liu2022flow} supports numerical inversion for extraction. Extended background and solver comparisons appear in \cref{app:related_work,app:solver}.

\textbf{Unimodal watermarking.} Post-processing methods embed signals in decoded video or audio~\citep{luo2023dvmark,fernandez2024videoseal,chen2023wavmark,sanroman2024proactive}. In-processing methods such as Gaussian Shading~\citep{yang2024gaussian}, VideoShield~\citep{hu2025videoshield}, and VideoMark~\citep{hu2025videomark} instead use generative latents. SIGMark~\citep{zhu2026sigmark} addresses scalable blind extraction; SKeDA~\citep{yang2026skeda} improves temporal recovery through key permutation and differential attention. MarkDiffusion~\citep{pan2025markdiffusion} benchmarks generative watermarks. mAVE extends modality-level copyright attribution to session-level AV authentication.

\textbf{Trainable video watermarking.} LVMark~\citep{jang2024lvmark} learns latent-decoder modulation and watermark recovery; VIDSIG~\citep{huang2025vidsig} partially fine-tunes the latent decoder with temporal alignment. FlowMark~\citep{asnani2026flowmark} learns mask-guided, content-adaptive embedding. These methods optimize video attribution and integrity. mAVE introduces a complementary generation-session relation between audio and video through the unchanged joint generator, without watermark-specific training.

\textbf{Content-derived signals.} Audio fingerprints~\citep{wang2003fingerprint} and video hashes~\citep{facebook2019tmk} match reference media. SyncNet~\citep{chung2016syncnet} measures lip--speech alignment; ImageBind/CLAP~\citep{girdhar2023imagebind,laionclap2023} measure compatibility. AVFF~\citep{oorloff2024avff} learns AV correspondences for deepfake detection. mAVE adds a keyed generation-session decision to these retrieval and forensic signals.

\textbf{Cross-modal media watermarking.} V$^2$A-Mark~\citep{zhang2024v2amark} embeds localization/copyright signals into video frames and audio with trained watermarking modules, then combines the recovered copyright bits. \citet{kim2025crossmodal} use trainable invertible networks to hide audio in visual frames for audio recovery and tamper localization. Their media-domain copyright and localization functions complement mAVE's native session constraint in a joint diffusion transformer (DiT). Appendix~\ref{app:trainable_comparison} compares their training and verification objectives.

\textbf{Cryptographic protocols.} \citet{thietke2025cryptography} formalize sampling and key-management requirements. Gaussian Shading++~\citep{yang2025gaussianshadingpp} develops dual-channel seed recovery and signature-based third-party verification for image watermarks. mAVE addresses a complementary objective: provider-assisted authentication of a joint AV generation session (\cref{app:related_work}).

\textbf{Provenance infrastructure.} C2PA Content Credentials~\citep{c2pa2026spec} combine signed manifests with asset bindings, including watermark/fingerprint soft bindings. mAVE's embedded session signal can connect AV authentication to such provenance records; \cref{app:provenance_ecosystem} outlines the integration.

\section{Threat Model and Problem Formulation}
\label{sec:threat_model}

\textbf{Provider-assisted verification.} The provider controls latent initialization and retains session records. Each inference request has public index $I$ and secret $K_{priv}=m$. External callers submit verification requests through the provider's interface without receiving these secrets. The decision authenticates shared generation-session provenance; \cref{sec:detection} specifies the prompt and record-retrieval workflow.

The adversary has generator weights and generation/detection API access, but no $K_{priv}$~\citep{tolosana2020deepfakes,mirsky2021creation}. It can request outputs and adaptively choose or splice them. Detection exposes only accept/reject, never the secret score or gradient. We consider:
\textbf{(1) Removal:} erasing a modality's watermark through compression, blur, or noise.
\textbf{(2) Swap:} pairing video $\mathbf x_v^{(A)}$ and audio $\mathbf x_a^{(B)}$ from distinct sessions $A\ne B$, forming $\tilde{\mathbf x}=(\mathbf x_v^{(A)},\mathbf x_a^{(B)})$. Individual watermark validity and perceptual compatibility do not establish shared provenance. \cref{sec:security} analyzes adaptive selection and local perturbations within the allowed swap family.

\textbf{Authentication objective.} Under $H_0$, both modalities originate from the indexed session; under $H_1$, they come from independent sessions. We seek low false acceptance under $H_1$ and high recall for distorted $H_0$ samples, accounting for inversion errors.

\begin{figure}[t]
    \centering
    \includegraphics[width=0.9\textwidth]{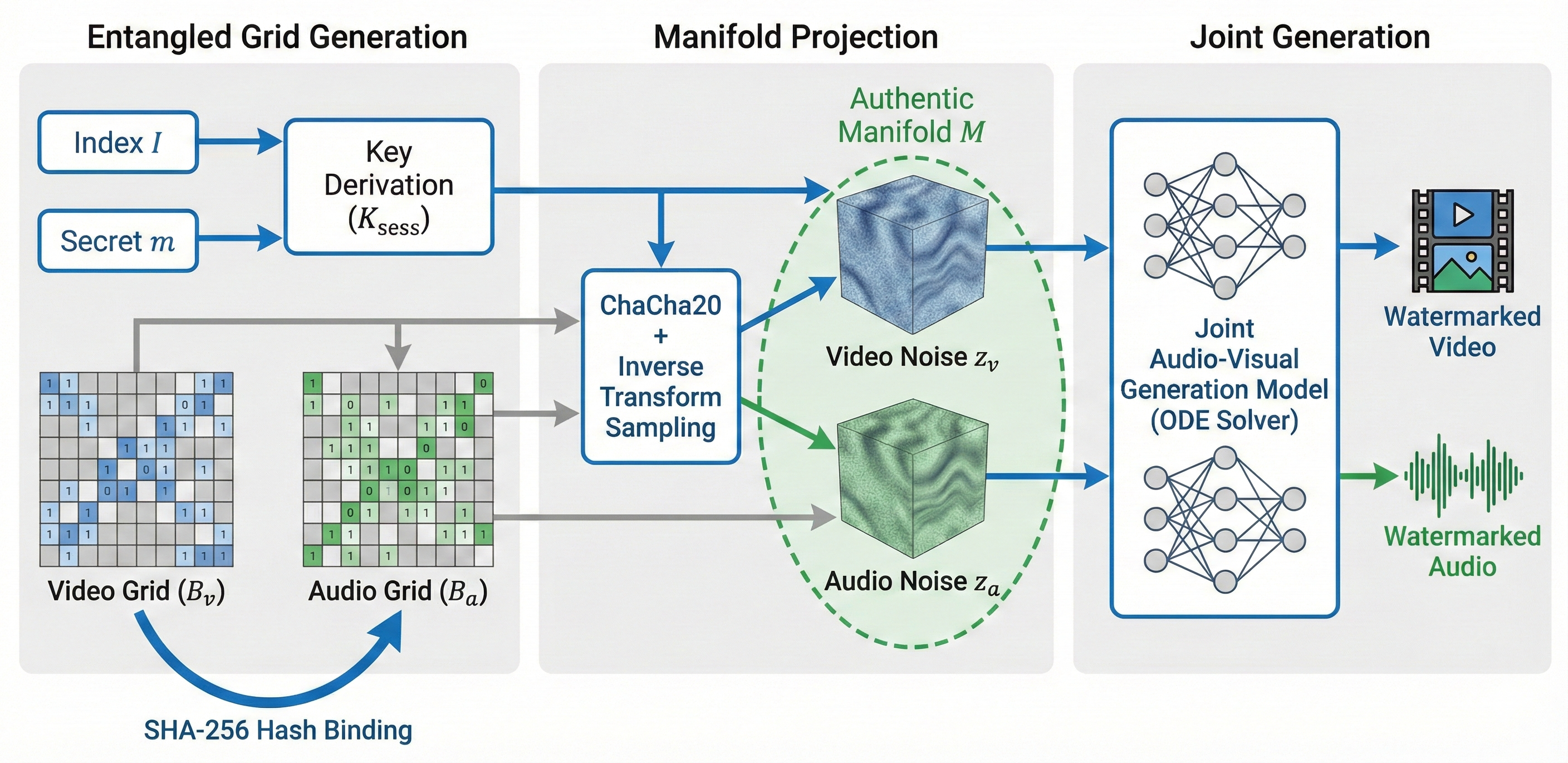}
    \caption{\textbf{mAVE pipeline.} A digest binds the audio grid to the video grid. Repetition, keyed randomization, and inverse-CDF sampling map these grids to joint initial latents, which the unchanged AV model denoises.}
    \label{fig:method}
\end{figure}

\section{Method}
\label{sec:method}

mAVE constructs bound grids, embeds them in joint latents, and verifies them after inversion (\cref{sec:manifold_construction,sec:embedding,sec:detection}); \cref{app:notation} collects notation.

\subsection{Constructing the Authentic Manifold}
\label{sec:manifold_construction}

Let $\mathcal{Z}_v,\mathcal{Z}_a$ be the video/audio initial latent spaces, $\mathcal{Z}=\mathcal{Z}_v\times\mathcal{Z}_a$, and $(\mathbf{z}_v,\mathbf{z}_a)$ a latent pair. The session key $K_{sess}$ defines an Authentic Manifold $\mathcal{M}$ through a binding predicate $\operatorname{Verify}$ (\cref{sec:detection}):
\begin{equation}
    \mathcal{M} = \{ (\mathbf{z}_v, \mathbf{z}_a) \in \mathcal{Z} \mid \text{Verify}(\mathbf{z}_a, \text{SHA\text{-}256}(\mathbf{z}_v, K_{sess})) = \text{True} \}.
\end{equation}
Here the latent-level notation summarizes the bit-grid construction below; hashes operate on discrete grids. The $N=|\mathcal{I}_{bind}|$ binding slots form a sparse constraint relative to $\dim(\mathcal{Z})$.

\textbf{Session-Aware Key Derivation.} Let $L_I,L_m$ be the bit lengths of public index $I\in\{0,1\}^{L_I}$ and secret $m\in\{0,1\}^{L_m}$. Only $I$ is embedded; $m=K_{priv}$ remains on the server. For original prompt $P$, let $E_P=\operatorname{Norm}(P)$ denote lowercase conversion and whitespace trimming. With $\parallel$ denoting concatenation and $\operatorname{Prefix}_{\lambda_{key}}$ retaining a fixed $\lambda_{key}$ hash bits independently of binding length $N$, derive
\begin{equation}
\label{eq: key}
    K_{sess} = \text{SHA-256}(\text{Prefix}_{\lambda_{key}}(\text{SHA-256}(m)) \parallel \text{SHA-256}(E_P)).
\end{equation}
Purpose-separated subkeys $K_{grid},K_{mask},K_{perm}$ serve grid construction, masking, and block permutation (\cref{app:key_derivation}). HMAC (hash-based message authentication code) derives modality sub-keys as $\operatorname{HMAC}(K_{grid},\text{modality})$; shared base entropy is seeded from $m$ (\cref{app:key_derivation}).

\textbf{Entangled Bit Grid Generation}
We construct one-dimensional watermark grids $B_v,B_a\in\{0,1\}^L$ of length $L$, subsequently reshaped and repeated into the audio/video latent tensors, incorporating \textit{Time-Templates} and \textit{Hash-Binding} (see \cref{fig:method}, Left).

\textbf{1. Video Grid ($B_v$).}
The grid has three disjoint regions: (i) a \textbf{Time-Template} $\mathbf{t}$ (channel $c_0$) for synchronization, (ii) the index $I$ in fixed, unencrypted repetition slots, and (iii) HMAC-derived base bits. Here $\mathcal{I}_{inject}$ inserts the template and index into designated slots, $S_{shared}$ is the shared pseudorandom bit stream, and $\oplus$ is bitwise XOR:
\begin{equation}
    B_v = \mathcal{I}_{inject}\!\left( \text{HMAC}(K_{sess}, \text{``video''}) \oplus S_{shared},\  \mathbf{t},\  I \right).
\end{equation}

\textbf{2. Audio Grid with Binding ($B_a$).}
To enforce dependency, we embed a digest of the video bits, $h_v = \text{SHA-256}(B_v)$, at fixed logical binding slots $\mathcal{I}_{bind}$ (128 distinct digest bits by default). Let $i$ index a grid bit, $\mathcal{I}_{time}$ contain template slots, $\phi$ select distinct digest bits, and $\operatorname{Base}_a$ denote the audio sub-key-derived base grid:
\begin{equation}
    B_a[i] = 
    \begin{cases} 
    \mathbf{t}[i] & \text{if } i \in \mathcal{I}_{time} \quad \text{(Sync)}, \\
    h_v[\phi(i)] & \text{if } i \in \mathcal{I}_{bind} \quad \text{(Binding)}, \\
    \text{Base}_a[i] & \text{otherwise.}
    \end{cases}
    \label{eq:audio_binding}
\end{equation}
The resulting $\mathcal{M}$ is a union of hash-selected latent orthants, rather than a smooth lower-dimensional surface. We denote the randomized inverse-CDF projection from grids to Gaussian latents by $\Pi$.

\subsection{Embed: Inverse Transform Sampling on Manifolds}
\label{sec:embedding}

A critical challenge in latent watermarking is mapping a low-dimensional payload to a high-dimensional latent tensor $\mathbf{z} \in \mathbb{R}^{C \times T \times H \times W}$ with channel, temporal, height, and width sizes $C,T,H,W$. We use three stages: \textit{Diffusion}, \textit{Randomization}, and \textit{Inverse Sampling} (\cref{fig:method}, Middle).

\textbf{Watermark Diffusion and Randomization}
Let $B \in \{0,1\}^L$ be the entangled bit grid constructed in \cref{sec:manifold_construction}. We first apply \textbf{Watermark Diffusion}. The discrete watermark bits are replicated block-wise to form a diffused tensor $B^{diff}$ matching the target latent dimensions.

\textbf{Watermark Randomization} masks repetition patterns in the authentication payload while keeping the routing index recoverable. Using the masking subkey $K_{mask}$, we generate a pseudo-random keystream $\mathcal{K}$ via the ChaCha20~\citep{bernstein2008chacha} stream cipher. Let $\mathcal J_I$ be the fixed public index-copy coordinates and $\mathcal J_R$ their complement (all audio coordinates lie in $\mathcal J_R$). A secret uniform block-level permutation acts on $\mathcal J_R$, followed by independent per-copy masks; $M_{rand}$ is:
\begin{equation}
    M_{rand}[j] = \begin{cases}
B^{diff}[j], & j\in\mathcal J_I,\\
B^{diff}[j]\oplus\mathcal K[j], & j\in\mathcal J_R.
\end{cases}
\end{equation}
Randomization makes $M_{rand}|_{\mathcal J_R}$ pseudorandom, while the fixed public copies support direct index recovery~\citep{thietke2025cryptography}.

\textbf{Inverse Transform Sampling}
We map the binary stream $M_{rand}$ to the continuous Gaussian latent $\mathbf{z}$ using the Inverse Probability Integral Transform.
Let $\Phi(\cdot)$ be the CDF of the standard normal distribution $\mathcal{N}(0, 1)$, and let $\text{ppf} = \Phi^{-1}$. For a standard uniform variable $u_i \sim \mathcal{U}(0, 1)$, the latent value $z_i$ is sampled as:
\begin{equation}
    z_i = \text{ppf}\left( \frac{u_i + M_{rand}[i]}{2} \right).
    \label{eq:sampling_formula}
\end{equation}
This effectively maps bit 0 to the negative half-Gaussian and bit 1 to the positive half-Gaussian.

\textbf{Implementation-aware prior preservation.} The sampler preserves the randomized payload distribution through the complete numerical pipeline, including clipping and finite precision. The reference uses the same public routing information and numerical operations, with independent uniform authentication signs.
\begin{theorem}[Implementation-matched distribution preservation]
\label{thm:lossless}
Let $\Gamma$ contain the public index, prompt, and sampler configuration, and let $M_R=M_{rand}|_{\mathcal J_R}$. Suppose the joint distinguishing advantage between $(\Gamma,M_R)$ and $(\Gamma,R)$ is at most $\epsilon_{rand}$, where $R$ is an independent uniform bit string of the same length. For the common numerical sampler, optionally followed by the joint generator, let $Y^s$ and $Y^{ref}$ denote its outputs under these respective inputs. Every efficient tester $\mathcal A$ satisfies
\begin{equation}
 \left|\Pr[\mathcal A(\Gamma,Y^s)=1]-\Pr[\mathcal A(\Gamma,Y^{ref})=1]\right|
 \le\epsilon_{rand}.
 \label{eq:implementation_preservation}
\end{equation}
\end{theorem}
\begin{proof}
A tester for the pipeline outputs induces a bit-stream tester by applying the same numerical sampler and generator to its challenge, using independent sampling coins. The two distinguishing advantages coincide. In the ideal independent-mask experiment, XOR with any prepared payload yields uniform authentication signs, so the two pipeline distributions agree exactly. Purpose-separated key derivation and secure masking supply the computational replacement (\cref{app:losslessness}).
\end{proof}
\textbf{Full-initialization budget.} Let $D$ be the total latent dimension and $\eta_{AV}$ the public-coordinate fraction. For the ideal-mask, continuously evaluated clamped reference $P_\delta^{ref}$, \cref{app:losslessness} proves
\begin{equation}
 \frac{W_2(P_\delta^{ref},\mathcal N(0,I_D))}{\sqrt D}
 \le\sqrt{2\eta_{AV}+(1-\eta_{AV})\chi_\delta},
 \label{eq:full_prior_budget}
\end{equation}
where $W_2$ is quadratic Wasserstein distance and $\chi_\delta$ is the per-coordinate mean-square displacement of the implemented clamp. Thus public routing and clipping receive an explicit distributional budget.

\subsection{Detect: Joint Inversion \& Verification}
\label{sec:detection}

\textbf{Verification workflow.} A caller supplies the AV pair and either its original prompt $P$ or generation-request metadata from which the provider retrieves $P$. This supplies condition $c$ before inversion. The provider then recovers $\hat I$ from fixed public slots, retrieves $m$, derives $K_{sess}$ using $E_P$, and verifies the payload and binding. Session secrets remain server-side; the interface returns the authentication decision (\cref{app:index_protocol}).

\textbf{Joint Flow Inversion}
Let $\mathbf{x}_{query}$ be the queried AV pair in the model's joint data-latent representation, $\mathbf{z}_t$ its flow state at time $t$, and $v_\theta$ the velocity field with model parameters $\theta$. One backward ODE pass from $t=1$ to the minimum timestep $\delta_t>0$ estimates the initial latent $\tilde{\mathbf{z}}_0$:
\begin{equation}
    \tilde{\mathbf{z}}_0 = \mathbf{x}_{query} + \int_{1}^{\delta_t} v_\theta(\mathbf{z}_t, t, c)  dt.
\end{equation}
Five-step inversion retains strong extraction performance in our ablation; solver details are in \cref{app:solver}. Joint inversion takes 2.07\,s/item versus 4.16\,s/item for two separate inversions on four RTX 5090 GPUs, with matching BA and reported peak/average VRAM of 32/22\,GB (\cref{app:joint_vs_separate}).

\textbf{Decoding and verification.} Zero-thresholding, key-dependent derandomization, and majority voting recover the repeated bits. After inversion, majority voting on $\mathcal J_I$ recovers $\hat I$; lookup yields $m$ and hence $K_{sess}$. Only then is $\mathcal J_R$ derandomized (\cref{app:index_protocol}). The provider reconstructs the noise-free video grid and hashes it to obtain the stable reference $H_{ideal}$. The score $S_{bind}$ is the fraction of recovered audio binding bits matching $H_{ideal}$ (\cref{eq:app_binding_score}). For query $\mathbf{x}$, detector $\mathcal{D}$ accepts when video/audio bit accuracies $\mathrm{BA}_v,\mathrm{BA}_a$ exceed $\tau_{acc}$ and binding score exceeds $\tau_{bind}$:
\begin{equation}
\mathcal D(\mathbf x)=(\mathrm{BA}_v>\tau_{acc})\land(\mathrm{BA}_a>\tau_{acc})\land(S_{bind}>\tau_{bind}).
\label{eq:and_gate}
\end{equation}
High mean BA alone does not determine this sample-level decision. Decoding and local temporal perturbations are examined in \cref{app:decoding,app:temporal}.

\subsection{Adaptive Binding Security}
\label{sec:security}

For a fixed pool of $n_s$ sessions, let $Q_{ab}=(\mathbf x_v^{(a)},\mathbf x_a^{(b)})$, $a\ne b$, be an ordered swap with fixed prompt and metadata. In the ideal recovered-bit model, each fixed pair's decoded $N$-bit audio binding vector is uniform conditional on its target record. For $1/2<\tau_{bind}<1$, define minimum match count $\kappa=\lfloor N\tau_{bind}\rfloor+1$, mismatch radius $r_\tau=N-\kappa$, and the Hamming-ball probability
\begin{equation}
 \varepsilon_N(b)=2^{-N}\!\sum_{j=0}^{\min\{N,r_\tau+b\}}\binom Nj,
 \qquad b\in\{0,\ldots,N\}.
 \label{eq:binding_ball}
\end{equation}

\begin{theorem}[Adaptive session-pool security]
\label{thm:security}
Let $M=n_s(n_s-1)$. Suppose that, except on an event of probability $\beta_g$, all allowed local perturbations of all $M$ swaps preserve their target records and change at most $b$ recovered binding bits. For adaptive search within these neighborhoods, the false-acceptance event $\mathrm{Forge}$ obeys
\begin{equation}
 \Pr[\mathrm{Forge}]\le\min\{1,M\varepsilon_N(b)+\beta_g\}.
 \label{eq:adaptive_pool_bound}
\end{equation}
Candidate choices and verification outcomes may be arbitrarily correlated.
\end{theorem}

\begin{proof}
For any fixed pair, let $U_0,U_\delta$ be its baseline and perturbed decoded binding vectors and $h$ its target digest. On the stability event, acceptance implies $d_H(U_\delta,h)\le r_\tau$ and $d_H(U_\delta,U_0)\le b$, hence $d_H(U_0,h)\le r_\tau+b$. Here $d_H$ is Hamming distance. The baseline uniform law gives probability $\varepsilon_N(b)$ for this enlarged ball. Union over the $M$ fixed pairs and add the global failure probability $\beta_g$. Every successful adaptive search must find a candidate in this union.
\end{proof}

\textbf{Consequences.} Pure swaps take $b=0$, $\beta_g=0$. For $N=128$ and $\tau_{bind}=0.8$, the exact single-test tail is $\varepsilon_N(0)\approx9.67\times10^{-13}$; the bound for all ordered swaps from $n_s=1{,}000$ sessions is $9.66\times10^{-7}$. \cref{app:binding_security} proves a bound $\min\{1,q\varepsilon_{single}\}$ for $q$ binary attack queries, where $\varepsilon_{single}$ bounds the best no-verification-feedback constructor, and derives decoding-margin certificates connecting local input changes to $b$.

\section{Experiments}
\label{sec:experiments}

\subsection{Implementation Details.}
\label{sec:implementation}
We evaluate LTX-2~\citep{hacohen2026ltx2} and MOVA-720p~\citep{openmossteam2026mova}. Clips are 5\,s at $256\times256$ and 24 fps. The $L=512$ payload bits span the full latent tensor, with channel/time/height/width repetition factors $(k_c,k_t,k_h,k_w)=(3,1,4,4)$ and a 32-bit public index $I$. Audio grids are reshaped/repeated to their latent dimensions. I2AV images use Stable Diffusion~\citep{rombach2022high}. Key derivation, ChaCha20 randomization~\citep{bernstein2008chacha}, and clipped inverse sampling are detailed in \cref{app:key_derivation}.

\textbf{Evaluation Prompts.}
We curate 250 text prompts from VBench~\citep{huang2024vbench} spanning 8 semantic categories. Each prompt is paired with 4 random seeds, yielding 1{,}000 audio-visual samples.

\textbf{Evaluation Protocol and Metrics.} All detection thresholds are selected once on a held-out clean validation split and then frozen for swap and distortion tests, without attack-specific tuning. Extraction BA measures correct bits; watermark detection measures samples passing the BA gates; final authentication additionally checks binding (\cref{eq:and_gate}). True-positive/false-negative rates (TPR/FNR) measure acceptance/rejection of same-session pairs; true-negative/false-positive rates (TNR/FPR) measure rejection/acceptance of swaps. Each swap group in \cref{tab:swap_defense,tab:binding_controls} contains 1{,}000 independent trials. mAVE records 0/1{,}000 false acceptances per group (one-sided 95\% FPR upper limit: 0.30\%).

\subsection{Extraction Performance}
\label{sec:exp_extraction}

\cref{tab:extraction} reports majority-vote payload BA for mAVE and unimodal baselines~\citep{chen2023wavmark,fernandez2024videoseal,hu2025videoshield,liu2024timbre,sanroman2024proactive}; AudioSeal instead reports a detection score. Tasks are text-to-AV (T2AV), image-to-AV (I2AV), and text-and-image-to-AV (TI2AV). BA below 1 does not imply failed detection: a sample passes when its BA exceeds the threshold. Detection rates require per-sample counts (\cref{app:metrics}).

\begin{table}[ht]
\centering
\caption{\textbf{Watermark Detection and Extraction.} Det. V/A = mAVE video/audio detection rates (\%). VS = VideoShield; VSeal = VideoSeal; AS = AudioSeal score. Remaining entries are extraction BA, not detection rates.}
\label{tab:extraction}
\small
\setlength{\tabcolsep}{3pt}
\begin{tabular}{@{}lc|ccc|cccc@{}}
\toprule
& mAVE & \multicolumn{3}{c|}{Video BA} & \multicolumn{4}{c}{Audio BA / score} \\
Model / task & Det. V/A & VS & VSeal & mAVE & AS & WavMark & Timbre & mAVE \\
\midrule
LTX-2 / T2AV & 100/100 & .953 & 1.000 & .936 & 1.000 & .994 & .993 & .915 \\
LTX-2 / I2AV & 100/99.8 & .946 & .995 & .934 & 1.000 & .995 & 1.000 & .917 \\
MOVA-720p / TI2AV & 100/99.9 & .952 & .998 & .949 & 1.000 & .993 & .999 & .928 \\
\bottomrule
\end{tabular}
\end{table}

\subsection{Generation Fidelity}
\label{sec:exp_fidelity}

We evaluate deployed-system quality with both public routing and randomized authentication present. This complements the implementation-matched guarantee and full-prior budget in \cref{thm:lossless,eq:full_prior_budget}.
We employ VBench metrics for video quality: \textit{Subject Consistency} (Sub), \textit{Motion Smoothness} (Mot), \textit{Dynamic Degree} (Dyn), \textit{Background Consistency} (Back), and \textit{Imaging Quality} (Img). For audio-visual quality, we use \textbf{CLAP Score}~\citep{laionclap2023} (Audio-Text relevance) and \textbf{ImageBind}~\citep{girdhar2023imagebind} (AV alignment). Synchronization is measured via \textbf{SyncNet}~\citep{chung2016syncnet} confidence.

\begin{table}[h]
\centering
\caption{\textbf{Fidelity Comparison.} Video quality, audio--text relevance, and AV alignment. Avg averages the five video-quality metrics.}
\label{tab:fidelity}
\setlength{\tabcolsep}{4pt}
\resizebox{0.9\textwidth}{!}{%
\begin{tabular}{l|cccccc|cc|c}
\toprule
\textbf{Method} & \textbf{Sub} $\uparrow$ & \textbf{Mot} $\uparrow$ & \textbf{Dyn} $\uparrow$ & \textbf{Back} $\uparrow$ & \textbf{Img} $\uparrow$ & \textbf{Avg} & \textbf{CLAP} $\uparrow$ & \textbf{ImBind} $\uparrow$ & \textbf{Sync} $\uparrow$ \\ \midrule
Clean (No WM)   & 0.983 & 0.988 & \textbf{0.715} & 0.966 & 0.524 & \textbf{0.835} & \textbf{0.442} & 0.117 & 0.965 \\
Uncoupled (Base)& 0.994 & 0.988 & 0.684 & 0.982 & 0.454 & 0.820 & 0.425 & 0.098 & 0.965 \\
\textbf{mAVE (Ours)} & \textbf{0.998} & \textbf{0.991} & 0.669 & \textbf{0.985} & \textbf{0.529} & 0.834 & 0.426 & \textbf{0.132} & \textbf{0.966} \\ \bottomrule
\end{tabular}%
}
\end{table}

\textbf{Results.} \cref{tab:fidelity} shows comparable quality, including synchronization scores of 0.966 and 0.965. A blind mean-opinion-score study ($N=15$) reports clean/watermarked visual means of 4.20/4.27 and audio means of 3.47/3.53. The automatic metrics and human ratings indicate comparable generation quality.

\subsection{Binding Security}

We compare independent VideoShield + AudioSeal, consistency checks, and Shared Session-ID (decoded-tag agreement). To separate hash dependency from access to a noise-free reference, we add server-referenced random binding (SR-Rand) and shared ID (SR-ID). They retain the embedding budget and frozen calibration rule; SR-Rand also retains mAVE's binding gate (\cref{app:binding_controls}). 

\begin{table}[h]
\centering
\caption{\textbf{Watermark Detection versus Pair Authentication.} WM det. is the authentic-pair rate passing both modality detectors before consistency/binding checks. Each method uses 1{,}000 independently sampled swaps. All rates use frozen thresholds. Acc = $(\mathrm{TPR}+\mathrm{TNR})/2$.}
\label{tab:swap_defense}
\setlength{\tabcolsep}{6pt}
\resizebox{0.95\linewidth}{!}{
\begin{tabular}{lcccccc}
\toprule
Method & WM det. $\uparrow$ & TPR (Auth) $\uparrow$ & FNR (Auth) $\downarrow$ & TNR (Swap) $\uparrow$ & FPR (Swap) $\downarrow$ & Acc $\uparrow$ \\
\midrule
Weak Base (Uncoupled) & 100\% & \textbf{100\%} & \textbf{0\%} & 0\% & 100\% & 50.0\% \\
Uncoupled + SyncNet & 100\% & 96.2\% & 3.8\% & 76.2\% & 23.8\% & 86.2\% \\
Uncoupled + ImageBind/CLAP & 100\% & 97.8\% & 2.2\% & 65.4\% & 34.6\% & 81.6\% \\
Shared Session-ID & 99.8\% & 98.9\% & 1.1\% & 88.7\% & 11.3\% & 93.8\% \\
SR-Rand & 99.7\% & 99.7\% & 0.3\% & 93.2\% & 6.8\% & 96.5\% \\
SR-ID & 99.7\% & 99.7\% & 0.3\% & 95.7\% & 4.3\% & 97.7\% \\
mAVE (Ours) & 99.8\% & 99.8\% & 0.2\% & \textbf{100\%} & \textbf{0\%} & \textbf{99.9\%} \\
\bottomrule
\end{tabular}
}
\end{table}

\cref{tab:swap_defense} reports 6.8\%/4.3\% FPR for SR-Rand/SR-ID at 99.7\% TPR, versus mAVE's 0\% FPR at 99.8\% TPR. \cref{fig:security_analysis} shows the original configurations.

\begin{figure}[ht]
  \centering
  \includegraphics[width=\textwidth]{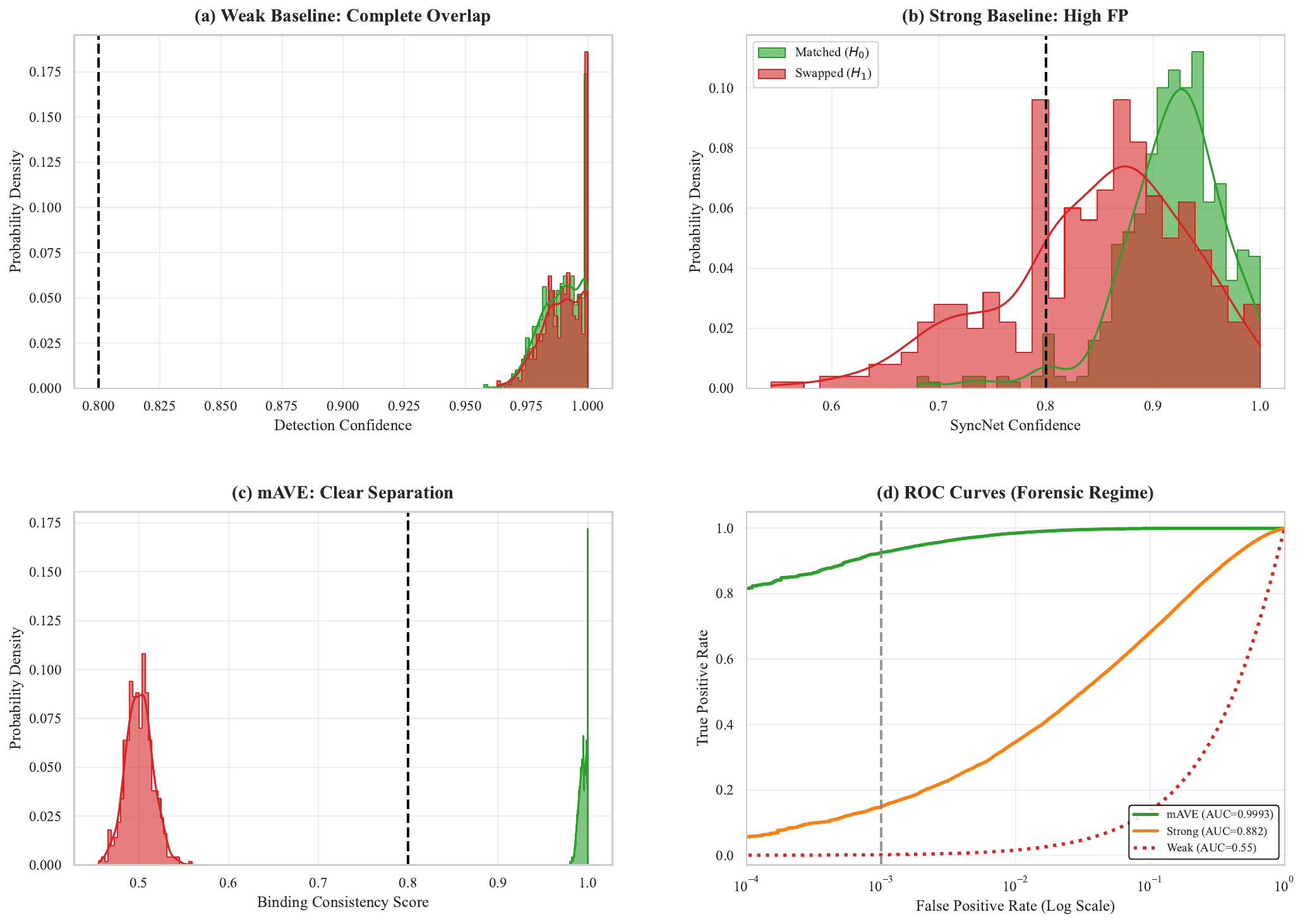}
  \caption{\textbf{Binding-security analysis.} (a) Independent watermarks overlap across authentic and swapped pairs. (b) SyncNet leaves residual overlap. (c) mAVE separates the evaluated pairs. (d) ROC curves compare the three original configurations; ImageBind/CLAP results are in \cref{tab:swap_defense}.}
  \label{fig:security_analysis}
\end{figure}

\subsection{Robustness}
\label{sec:exp_robustness}

\cref{tab:joint_robustness} reports extraction and final authentication. mAVE performs well under temporal attacks: FrameAvg achieves 0.927 video BA, 99.2\% TPR, and 0\% observed FPR. Across the tested attacks, TPR ranges from 97.4\% to 99.5\%; MP4 and regeneration retain 99.0\% and 98.1\% TPR, respectively. Attack protocols and further temporal results appear in \cref{app:regeneration,app:temporal}.

\begin{table}[ht]
\centering
\caption{\textbf{Robustness by Attack Domain.} Extraction reports BA, except AudioSeal (AS) scores; VS = VideoShield. TPR/FPR measure full mAVE pair authentication (\%) at frozen thresholds. FNR $=100-\mathrm{TPR}$ and TNR $=100-\mathrm{FPR}$.}
\label{tab:joint_robustness}
\begingroup
\small
\setlength{\tabcolsep}{2pt}
\renewcommand{\arraystretch}{1.05}
\begin{minipage}[t]{0.49\linewidth}
\centering
\textbf{(a) Video attacks}\\[2pt]
\begin{tabular}{@{}lcccc@{}}
\toprule
& \multicolumn{2}{c}{Extraction} & \multicolumn{2}{c}{Auth. (\%)}\\
Attack / parameter & VS & mAVE & TPR $\uparrow$ & FPR $\downarrow$\\
\midrule
FrameAvg ($N=3$) & .948 & .927 & 99.2 & 0\\
Noise ($\sigma=.1$) & .907 & .892 & 98.9 & 0\\
G-Blur ($r=2$) & .947 & .939 & 99.5 & 0\\
M-Blur ($k=3$) & .952 & .938 & 98.6 & 0\\
S\&P ($p=.1$) & .863 & .855 & 97.7 & 0\\
Resize ($.3\times$) & .951 & .936 & 99.4 & 0\\
Brightness ($f=3$) & .943 & .914 & 99.3 & 0\\
\bottomrule
\end{tabular}
\end{minipage}\hfill
\begin{minipage}[t]{0.49\linewidth}
\centering
\textbf{(b) Audio attacks}\\[2pt]
\begin{tabular}{@{}lcccc@{}}
\toprule
& \multicolumn{2}{c}{Extraction} & \multicolumn{2}{c}{Auth. (\%)}\\
Attack / parameter & AS & mAVE & TPR $\uparrow$ & FPR $\downarrow$\\
\midrule
MP3 (64k) & .97 & .85 & 98.5 & 0.1\\
Compression & 1.00 & .89 & 99.1 & 0\\
Lowpass (8kHz) & 1.00 & .86 & 98.9 & 0\\
Resample (16k$\to$8k) & 1.00 & .88 & 99.0 & 0\\
Noise (20dB) & .98 & .80 & 97.4 & 0.2\\
Quantization (8-bit) & .92 & .83 & 99.1 & 0\\
\bottomrule
\end{tabular}
\end{minipage}

\medskip
\textbf{(c) Joint delivery / regeneration}\\[2pt]
\begin{tabular*}{\linewidth}{@{\extracolsep{\fill}}lcccccc@{}}
\toprule
& \multicolumn{2}{c}{Video extraction} & \multicolumn{2}{c}{Audio extraction} & \multicolumn{2}{c}{Auth. (\%)}\\
Attack & VS & mAVE & AS & mAVE & TPR $\uparrow$ & FPR $\downarrow$\\
\midrule
MP4 (H.264 + AAC) & .932 & .940 & 1.00 & .942 & 99.0 & 0\\
Regeneration$^\dagger$ & .867 & .859 & .979 & .871 & 98.1 & 0\\
\bottomrule
\end{tabular*}
\par\smallskip
{\footnotesize $^\dagger$ Diffusion purification + EnCodec reconstruction.}
\endgroup
\end{table}

\subsection{Ablations}

\textbf{Parameter ablations.} For $N\in\{16,32,64,128,256\}$, video BA exceeds 95\% and audio BA is 87--91\%. Five inversion steps retain over 95\% video and 85\% audio BA; one step yields 85.7\%/78.0\% (\cref{app:solver}).

\textbf{Binding and hard swaps.} Across three swap constructions, SR-Rand/SR-ID yield 6.6--7.1\%/8.1--9.2\% FPR at 99.7\% TPR, versus mAVE's 0\% FPR at 99.8\% TPR (\cref{tab:binding_controls}). Protocols and ablation definitions appear in \cref{app:binding_controls}.

\begin{table}[ht]
\centering
\caption{\textbf{Hard-swap Controls} (\%). 1{,}000 independent swaps per group; separate from \cref{tab:swap_defense}.}
\label{tab:binding_controls}
\small
\setlength{\tabcolsep}{4pt}
\begin{tabular}{@{}lcccc@{}}
\toprule
& Authentic & \multicolumn{3}{c}{Swap FPR (\%) $\downarrow$} \\
Method / control & TPR (\%) $\uparrow$ & Cross-prompt & Same-prompt & Selected \\
\midrule
Uncoupled + SyncNet & 96.2 & 23.8 & 29.5 & 33.1 \\
Uncoupled + ImageBind/CLAP & 97.8 & 34.6 & 33.7 & 41.8 \\
Shared Session-ID & 98.9 & 11.3 & 11.2 & 10.8 \\
mAVE: no binding gate & 99.8 & 0.1 & 2.4 & 3.8 \\
mAVE: no hash or binding gate & 99.8 & 0 & 6.7 & 5.9 \\
SR-Rand & 99.7 & 6.9 & 6.6 & 7.1 \\
SR-ID & 99.7 & 8.2 & 9.2 & 8.1 \\
mAVE (full) & 99.8 & 0 & 0 & 0 \\
\bottomrule
\end{tabular}
\end{table}
\FloatBarrier

\section{Conclusion}
\label{sec:conclusion}

mAVE advances copyright attribution from independently valid audio and video marks to authentication of their shared generation session. It binds the two modalities at latent initialization and verifies them through one joint inversion, without changing generator weights or training auxiliary watermark networks. Experiments on LTX-2 and MOVA demonstrate comparable generation quality, with 99.8\% TPR and 0\% observed FPR in the evaluated swap test. Strong authentication under FrameAvg, MP4 compression, and regeneration further supports its use beyond unmodified outputs. Together, these results establish session-level pairing as a practical protection objective for native joint audio-visual generation.

\clearpage
\subsection*{AI use statement}
Generative AI tools assisted with language polishing, manuscript organization, literature discovery and summarization, and feedback on experimental design and protocol descriptions, including the additional session-binding controls, clarification of the index protocol and theoretical scope, formalization and checking of adaptive-security proofs, and design of matched server-reference controls. The authors directed these revisions and retain responsibility for the final text, claims, references, and experimental protocols.

\subsection*{Ethics statement}
mAVE supports copyright protection and content provenance by authenticating the shared generation session of audio and video. It is intended to help providers and users identify misleading cross-session combinations and strengthen accountability for synthetic media. Responsible deployment should protect session secrets and prompt metadata, respect the licenses and usage conditions of underlying models and media, and communicate verification results accurately.

\subsection*{Reproducibility statement}
The main paper and appendices document the method, implementation choices, and evaluation protocols to support reproducibility. \Cref{sec:method,sec:experiments} describe the embedding and verification workflow, model configurations, and experimental setup. The appendices provide theoretical assumptions and full proofs (\cref{app:proofs}), inversion and decoding procedures (\cref{app:solver,app:decoding}), key derivation and encoding/verification pseudocode (\cref{app:key_derivation}), metric definitions (\cref{app:metrics}), prompt examples and attack protocols (\cref{app:temporal,app:prompts}), and session-binding controls (\cref{app:binding_controls}).

\clearpage
\bibliographystyle{iclr2027_conference}
\bibliography{main}

\clearpage
\appendix
\numberwithin{equation}{section}
\numberwithin{figure}{section}
\numberwithin{table}{section}
\section{Full Proofs of Theoretical Guarantees}
\label{app:proofs}
\subsection{\texorpdfstring{Implementation-Matched Distribution Preservation (\cref{thm:lossless})}{Implementation-Matched Distribution Preservation}}
\label{app:losslessness}

The argument separates three objects: the exact continuous sampling law, the complete numerical implementation with a matched reference, and the reference's distance from a standard Gaussian prior. This separates cryptographic randomization strength from the costs of public routing and numerical sampling.

\subsubsection{Exact Joint Law of Inverse-Transform Embedding}

Let $M\in\{0,1\}^d$ be any sign-bit vector, allowing arbitrary dependence between its coordinates. Independently draw $U_j\overset{\mathrm{iid}}\sim\mathcal U(0,1)$ and set $Z_j=\Phi^{-1}((U_j+M_j)/2)$. Write $s(z)_j=\mathbb1[z_j>0]$ and $\varphi_d(z)=\prod_j\varphi(z_j)$, where $\varphi$ is the standard Gaussian density. Denote the standard Gaussian law by $\gamma_d=\mathcal N(0,I_d)$.

\begin{lemma}[Joint sign-mixture identity]
\label{lem:dist_preserve}
The untruncated continuous sampler has joint density
\begin{equation}
 p_Z(z)=2^d\Pr[M=s(z)]\varphi_d(z)
 \label{eq:joint_sampling_density}
\end{equation}
almost everywhere. Moreover,
\begin{align}
 \operatorname{TV}(P_Z,\gamma_d)
 &=\operatorname{TV}(P_M,\operatorname{Unif}(\{0,1\}^d)),\label{eq:joint_tv_identity}\\
 D_{KL}(P_Z\Vert\gamma_d)&=d\ln2-H(M),\label{eq:joint_kl_identity}
\end{align}
where $\operatorname{TV}$ is total variation distance, $D_{KL}$ is relative entropy, and $H$ uses natural logarithms.
\end{lemma}
\begin{proof}
Given $M=m$, independent uniform sampling gives
\[
 p(z\mid m)=2^d\varphi_d(z)\mathbb1[s(z)=m].
\]
Summing over $m$ proves Eq.~\eqref{eq:joint_sampling_density}. Each orthant has Gaussian mass $2^{-d}$. The likelihood ratio is constant on that orthant and equals $2^d\Pr[M=m]$. Integrating the absolute density difference therefore yields
\[
 \tfrac12\int|p_Z-\varphi_d|
 =\tfrac12\sum_m|\Pr[M=m]-2^{-d}|.
\]
Similarly, integrating the log-likelihood ratio gives
\[
 D_{KL}(P_Z\Vert\gamma_d)
 =\sum_m\Pr[M=m]\ln(2^d\Pr[M=m])=d\ln2-H(M),
\]
with $0\ln0=0$. These are identities for the joint law, regardless of dependence among the bits.
\end{proof}

Independent uniform sign bits recover the complete Gaussian product law. Correlated signs are characterized by the same identities: for example, two equal uniform signs yield Gaussian marginals but joint KL divergence $\ln2$ and TV distance $1/2$. These are information-theoretic distribution statements; computational pseudorandomness is handled by the implementation reduction below.

\paragraph{General $l$-bit partitions.}
For $M\in\{0,\ldots,2^l-1\}^d$, use $Z_j=\Phi^{-1}((U_j+M_j)/2^l)$. Define $s_l(z)_j=\lfloor2^l\Phi(z_j)\rfloor$ away from the zero-probability endpoints. Conditional on symbol $i$, one coordinate has density
\begin{equation}
 p(z\mid M_j=i)=2^l\varphi(z)\,
 \mathbb1[\Phi^{-1}(i/2^l)<z\le\Phi^{-1}((i+1)/2^l)].
 \label{eq:app_conditional}
\end{equation}
Its CDF is 0 below the interval, 1 above it, and
\begin{equation}
 F(z\mid M_j=i)=2^l\Phi(z)-i
 \label{eq:app_cdf}
\end{equation}
inside it. Solving $u=F(z\mid M_j=i)$ gives
\begin{equation}
 z=\Phi^{-1}\!\left(\frac{u+i}{2^l}\right),\qquad u\sim\mathcal U(0,1),
 \label{eq:app_general_its}
\end{equation}
and the inverse bin map is
\begin{equation}
 i=\lfloor2^l\Phi(z)\rfloor.
 \label{eq:app_inverse_map}
\end{equation}
For the full vector, each bin product has Gaussian mass $2^{-ld}$, so the preceding proof gives $p_Z(z)=2^{ld}\Pr[M=s_l(z)]\varphi_d(z)$, the same TV identity relative to a uniform symbol vector, and $D_{KL}=dl\ln2-H(M)$. mAVE uses $l=1$.

\subsubsection{Complete Numerical Pipeline and Purpose Separation}

Let $\Gamma$ collect the index $I$, prompt $P$, fixed public routing layout, latent shapes, and numerical sampling configuration. The grid, permutation, and mask material are generated in separate purpose domains. Let $B^*$ be the repeated, permuted authentication payload in physical coordinate order, so $M_R=B^*\oplus\mathcal K_{mask}$. Public index coordinates remain unmasked. The payload may contain repeated bits, time templates, and any deterministic digest relation.

The cryptographic premise is joint pseudorandomization: replacing $M_R$ by an independent uniform vector $R$ while preserving $\Gamma$ changes the success probability of any efficient bit-stream distinguisher by at most $\epsilon_{rand}$. A sufficient hybrid first separates the mask-purpose key from the prepared grid/permutation state, with joint advantage $\epsilon_{sep}$, and then replaces its masking stream by uniform bits, with advantage $\epsilon_{stream}$. This yields $\epsilon_{rand}\le\epsilon_{sep}+\epsilon_{stream}$. The key-separation premise includes all state used to prepare the payload, including shared-seed material. Different application labels are the standard mechanism for separating KDF purposes~\citep{krawczyk2010hkdf}; masks for different physical copies and modalities occupy distinct stream positions~\citep{nir2018chacha}.

In the ideal hybrid, condition on all payload and permutation preparation. Independent uniform mask bits make $B^*\oplus R$ uniform, whatever the dependence within $B^*$. Thus this sampling result treats the binding digest as arbitrary payload data and does not require a random-oracle model for that digest.

Write $\mathcal S_{num}(\Gamma,M_R;\omega)$ for the actual full sampler, where $\omega$ denotes its independent random coins. It includes inverse-CDF evaluation, sample-then-clamp, finite random-number precision, floating-point operations, reshaping, and public-index insertion. The matched reference replaces only the authentication sign vector by $R$ and uses the identical numerical procedure and public header. A common generator conditioned on $P$, including its output decoder and any independent randomness, can be composed with $\mathcal S_{num}$.

\begin{proof}[Proof of \cref{thm:lossless}]
Suppose an efficient tester $\mathcal A$ distinguishes the complete pipeline outputs. Construct $\mathcal B$ that receives $(\Gamma,c)$, samples the pipeline's independent coins, and returns $\mathcal A(\Gamma,\mathcal P_{num}(\Gamma,c))$, where $\mathcal P_{num}$ is either the full sampler or the sampler followed by the generator. For $c=M_R$ this is the real experiment; for $c=R$ it is the matched-reference experiment. The distinguishing advantage is identical and hence at most $\epsilon_{rand}$. Clipping, rounding, and generation are common efficient postprocessing, so no exact-arithmetic assumption is needed. With truly independent uniform masks the two input laws agree, and so do the complete output laws.
\end{proof}

The information-theoretic reference in the remaining analysis uses ideal uniform masking. The real finite-key implementation is connected to that reference through Eq.~\eqref{eq:implementation_preservation}, rather than by identifying pseudorandom bits with statistically uniform bits.

\subsubsection{Exact Cost of Public Routing in the Continuous Prior}

Let $D=d_I+d_R$ be the full latent dimension, where $d_I=|\mathcal J_I|=L_Ir_I$ counts physical public-index copies and $d_R$ the complementary coordinates. Let $C(I)\in\{0,1\}^{d_I}$ be the injective repeated index code. In the ideal untruncated experiment the full sign vector is $(C(I),R)$, with $R$ independent uniform and independent of $I$. Therefore $H(C(I),R)=H(I)+d_R\ln2$, and \cref{lem:dist_preserve} gives
\begin{equation}
 D_{KL}(P_Z^{ref,0}\Vert\gamma_D)=d_I\ln2-H(I).
 \label{eq:routing_entropy_cost}
\end{equation}
For a fixed index this equals $d_I\ln2$. For an index uniform over all $2^{L_I}$ values, it equals $(d_I-L_I)\ln2$, and
\begin{equation}
 \operatorname{TV}(P_Z^{ref,0},\gamma_D)=1-2^{L_I-d_I}.
 \label{eq:routing_tv_cost}
\end{equation}
Indeed, the legal index patterns occupy $2^{L_I}$ of $2^{d_I}$ possible public sign patterns, with uniform mass on the legal set. Conditioning on a particular public index instead gives TV $1-2^{-d_I}$. These formulas explicitly account for routing redundancy in the full latent distribution. The next result measures its mean-square transport cost together with the implemented clamp.

\subsubsection{The Implemented Half-Interval Clamp}

The implemented rule in \cref{app:key_derivation} first draws $u\sim\mathcal U(1/2,1)$ and then clamps it to $[1/2+\delta,1-\delta]$, with $0<\delta<1/4$. Define
\[
 \ell_\delta=\Phi^{-1}(1/2+\delta),\qquad
 a_\delta=\Phi^{-1}(1-\delta),\qquad
 H_\delta(x)=\min\{a_\delta,\max\{\ell_\delta,x\}\},\quad x\ge0.
\]
For a uniform sign, the continuously evaluated output is distributed as $T_\delta(G)=\operatorname{sgn}(G)H_\delta(|G|)$, $G\sim\mathcal N(0,1)$. Thus both near-zero magnitudes and extreme tails are clamped. In one dimension the law agrees with the Gaussian density on $(-a_\delta,-\ell_\delta)\cup(\ell_\delta,a_\delta)$, and has mass $\delta$ at each of $-a_\delta,-\ell_\delta,\ell_\delta,a_\delta$.

For $d$ independent uniform-sign coordinates, their common continuous interior mass is $(1-4\delta)^d$, and the remaining clamped mass is supported on boundaries of zero Gaussian mass. Consequently,
\begin{equation}
 \operatorname{TV}(\mathcal L(T_\delta(G_1),\ldots,T_\delta(G_d)),\gamma_d)
 =1-(1-4\delta)^d.
 \label{eq:actual_clamp_tv}
\end{equation}
In this continuous comparison the forward KL divergence is infinite because of the boundary atoms. The implementation-matched theorem still applies exactly; mean-square displacement provides the useful quantitative comparison to the continuous prior.

\begin{lemma}[Clamping transport cost]
\label{lem:clamp_transport}
For the continuously evaluated clamp, define
\begin{align}
 \chi_\delta
 &=\mathbb E[(H_\delta(|G|)-|G|)^2]\notag\\
 &=2\!\int_0^{\ell_\delta}(\ell_\delta-x)^2\varphi(x)\,dx
   +2\!\int_{a_\delta}^{\infty}(x-a_\delta)^2\varphi(x)\,dx.\label{eq:clamp_transport_cost}
\end{align}
Then $W_2^2(\mathcal L(T_\delta(G_1),\ldots,T_\delta(G_d)),\gamma_d)=d\chi_\delta$. Also $v_\delta=\mathbb E[H_\delta(|G|)^2]\le1$.
\end{lemma}
\begin{proof}
The scalar map $T_\delta$ is nondecreasing, so coupling $G$ with $T_\delta(G)$ is the monotone optimal coupling for squared distance. Its cost is the two integrals in Eq.~\eqref{eq:clamp_transport_cost}. Coordinatewise coupling attains their sum for the product laws; every joint coupling has cost at least the sum of its optimal scalar marginal costs. Hence equality holds in $d$ dimensions.

For the second moment, using $\Pr[|G|\le\ell_\delta]=\Pr[|G|\ge a_\delta]=2\delta$ gives
\[
 v_\delta=1-4\delta+2\delta(\ell_\delta^2+a_\delta^2)
          +2\ell_\delta\varphi(\ell_\delta)-2a_\delta\varphi(a_\delta).
\]
Since $\ell'_\delta=1/\varphi(\ell_\delta)$ and $a'_\delta=-1/\varphi(a_\delta)$,
\[
 v'_\delta=4\delta\left[\frac{\ell_\delta}{\varphi(\ell_\delta)}
                         -\frac{a_\delta}{\varphi(a_\delta)}\right]\le0.
\]
The function $x/\varphi(x)$ is increasing for $x\ge0$, and $\ell_\delta<a_\delta$. Together with $v_\delta\to1$ as $\delta\downarrow0$, this proves $v_\delta\le1$.
\end{proof}

An explicit expression useful for checking the numerical value is
\begin{align}
 \chi_\delta={}&2[(1+\ell_\delta^2)\delta+
 \ell_\delta\varphi(\ell_\delta)-2\ell_\delta\varphi(0)]\notag\\
 &+2[(1+a_\delta^2)\delta-a_\delta\varphi(a_\delta)].
 \label{eq:clamp_cost_closed}
\end{align}
For the reported $\delta=0.001$, $\ell_\delta\approx0.00250663$, $a_\delta\approx3.09023231$, $\chi_\delta\approx2.88895\times10^{-4}$, and $v_\delta\approx0.996289$. The normalized authentication-region transport distance is therefore $\sqrt{\chi_\delta}\approx0.016997$. These are analytical sampler quantities.

\subsubsection{A Full-Initialization Routing and Clipping Budget}

Let $P_\delta^{ref}$ be the full ideal-mask reference law with the public index code and the preceding continuous clamp. Write $\eta_{AV}=d_I/D$ for the public-coordinate fraction.

\begin{lemma}[Full-initialization transport budget]
\label{lem:full_prior_budget}
For any distribution of the index independent of the sampling magnitudes,
\begin{equation}
 W_2^2(P_\delta^{ref},\gamma_D)
 \le d_I(1+v_\delta)+d_R\chi_\delta
 \le2d_I+d_R\chi_\delta.
 \label{eq:full_prior_transport}
\end{equation}
Normalizing by $D$ gives Eq.~\eqref{eq:full_prior_budget}.
\end{lemma}
\begin{proof}
Draw $G\sim\gamma_D$ independently of $I$. At authentication coordinates, set $Z_j=T_\delta(G_j)$; at public coordinates set $Z_j=(2C(I)_j-1)H_\delta(|G_j|)$. This constructs the required law, including any correlations among index copies. Authentication coordinates contribute $\chi_\delta$ each. At an index coordinate, symmetry of $G_j$ gives
\[
 \mathbb E[(Z_j-G_j)^2]
 =\mathbb E[H_\delta(|G_j|)^2]+\mathbb E[G_j^2]=v_\delta+1,
\]
since the cross term is odd in $G_j$ for either index sign. Summing these coupling costs bounds $W_2^2$ from above. Finally use $v_\delta\le1$. Any unused coordinates sampled directly from the standard Gaussian can instead be coupled identically at zero cost, so counting them in the complementary region only loosens the bound.
\end{proof}

This separates two concrete contributions: the physical public-routing fraction and the clamp's scalar transport cost. At the specified clamp level, Eq.~\eqref{eq:full_prior_budget} becomes
\[
 \frac{W_2(P_\delta^{ref},\gamma_D)}{\sqrt D}
 \le\sqrt{2\eta_{AV}+2.88895\times10^{-4}(1-\eta_{AV})}.
\]

\paragraph{Numerical evaluation and quality functionals.}
Let $P_{num}^{ref}$ be the numerical matched-reference law in Theorem~\ref{thm:lossless}. If a coupling to the continuously evaluated $P_\delta^{ref}$ has normalized RMS discrepancy at most $\zeta_{num}$, the Wasserstein triangle inequality gives
\begin{equation}
 \frac{W_2(P_{num}^{ref},\gamma_D)}{\sqrt D}
 \le\zeta_{num}+\sqrt{2\eta_{AV}+(1-\eta_{AV})\chi_\delta}.
 \label{eq:numerical_prior_budget}
\end{equation}
Here $\zeta_{num}$ accounts for finite random-number precision and numerical evaluation under the chosen coupling. Unlike this comparison to a continuous Gaussian, Theorem~\ref{thm:lossless} already uses identical finite-precision operations in both worlds.

For a bounded efficient output statistic $f\in[0,1]$, let $F(z)$ be its expectation after generation from $z$, at the fixed prompt. If $F$ is $L_F$-Lipschitz with respect to $\|z-z'\|_2/\sqrt D$, coupling and the computational comparison yield
\begin{equation}
 \left|\mathbb E f(X^s)-\mathbb E f(X_G)\right|
 \le\epsilon_{rand}+L_F\!\left[\zeta_{num}
       +\sqrt{2\eta_{AV}+(1-\eta_{AV})\chi_\delta}\right],
 \label{eq:quality_transfer_bound}
\end{equation}
where $X_G$ is generated from $\gamma_D$. Indeed, the same random-output test converts the bounded-statistic difference between real and matched-reference outputs to at most $\epsilon_{rand}$, while $|\mathbb EF(Z_{num}^{ref})-\mathbb EF(G)|\le L_F\mathbb E\|Z_{num}^{ref}-G\|_2/\sqrt D$ follows from the coupling. Apply Cauchy--Schwarz and Eq.~\eqref{eq:numerical_prior_budget}. This transfer statement parameterizes the generator/statistic sensitivity; the empirical quality results are reported directly in \cref{sec:exp_fidelity}.

\subsection{\texorpdfstring{Adaptive Session-Pool Security (\cref{thm:security})}{Adaptive Session-Pool Security}}
\label{app:binding_security}

The analysis has three complementary parts: exact single-test probabilities, adaptive search over a candidate family, and local perturbation stability. The first supplies a Hamming-ball probability; the latter two account for selection and modification without imposing independence between verification queries.

\subsubsection{Verification Experiment and Exact Single-Test Probability}

\textbf{Queries and references.} For an ordered cross-session pair, $Q_{ab}=(\mathbf x_v^{(a)},\mathbf x_a^{(b)})$ with $a\ne b$, fix the original video prompt, request metadata, preprocessing, inversion configuration, and detection thresholds. The verifier recovers the public index, retrieves its session record, and rebuilds the ideal video grid. Let $h\in\{0,1\}^N$ contain the $N$ distinct selected bits of its digest, and let $U\in\{0,1\}^N$ be the audio binding vector after inversion, derandomization, and majority decoding. The score is
\begin{equation}
 S_{bind}=\frac1N\sum_{i=1}^{N}\mathbb 1[U_i=h_i]
 =1-\frac{d_H(U,h)}N.
 \label{eq:app_binding_score}
\end{equation}
The server-held ideal reference is fixed for a selected record, independently of perturbations to the recovered video bits. Full acceptance also requires both modality BA gates in \cref{eq:and_gate}.

\textbf{Baseline recovered-bit model.} For every fixed indexed pair, the ideal experiment assumes
\begin{equation}
 \Pr[U=u\mid\text{target record}]=2^{-N},\qquad u\in\{0,1\}^N.
 \label{eq:baseline_uniform}
\end{equation}
This is the baseline law for the decoded vector. The cryptographic construction uses independent session secrets and the PRF/random-oracle abstraction; Eq.~\eqref{eq:baseline_uniform} specifies the recovered-bit model to which the statistical bounds apply. It is imposed once per fixed base pair, rather than anew after each adaptive query. Joint distributions across different pairs may be dependent.

\begin{lemma}[Exact binding acceptance]
\label{lem:independence}
For $1/2<\tau_{bind}<1$, let $\kappa=\lfloor N\tau_{bind}\rfloor+1$ and $r_\tau=N-\kappa$. Under Eq.~\eqref{eq:baseline_uniform},
\begin{equation}
 \Pr[S_{bind}>\tau_{bind}]=\varepsilon_N(0)
 =2^{-N}\sum_{j=0}^{r_\tau}\binom Nj,
 \qquad \Pr[\mathcal D(Q_{ab})=1]\le\varepsilon_N(0).
 \label{eq:exact_single_binding}
\end{equation}
\end{lemma}
\begin{proof}
Fix the target record and hence $h$. Uniformity of $U$ makes the mismatch vector $U\oplus h$ uniform, so its Hamming weight has distribution $\operatorname{Binomial}(N,1/2)$. The strict score threshold requires at least $\kappa$ matches, equivalently at most $r_\tau$ mismatches. There are $\sum_{j=0}^{r_\tau}\binom Nj$ such vectors among the $2^N$ equiprobable vectors. This proves the equality. The BA gates can only reduce the set of accepted candidates, giving the full-verifier inequality. Averaging over target records preserves the bound.
\end{proof}

Hoeffding's inequality~\citep{hoeffding1963probability} gives the convenient exponential form
\begin{equation}
 \varepsilon_N(0)\le\exp[-2N(\tau_{bind}-1/2)^2].
 \label{eq:app_hoeffding}
\end{equation}
At $N=128$, $\tau_{bind}=0.8$, the score needs 103 matches and allows 25 mismatches. Thus $\varepsilon_{128}(0)\approx9.6667\times10^{-13}$, while the Hoeffding upper bound is approximately $9.8595\times10^{-11}$. These quantities describe the ideal binding experiment; the reported empirical FPR is defined separately in \cref{app:metrics}.

\begin{table}[ht]
\centering\small
\caption{\textbf{Analytical Single-Test Probabilities.} Exact binding acceptance and its Hoeffding upper bound at $\tau_{bind}=0.8$.}
\label{tab:app_evasion}
\setlength{\tabcolsep}{4pt}
\begin{tabular}{@{}lccccc@{}}
\toprule
$N$ & 16 & 32 & 64 & 128 & 256\\
\midrule
$\varepsilon_N(0)$ & $1.06\!\times\!10^{-2}$ & $2.68\!\times\!10^{-4}$ & $2.28\!\times\!10^{-7}$ & $9.67\!\times\!10^{-13}$ & $2.33\!\times\!10^{-23}$\\
Hoeffding bound & $5.61\!\times\!10^{-2}$ & $3.15\!\times\!10^{-3}$ & $9.93\!\times\!10^{-6}$ & $9.86\!\times\!10^{-11}$ & $9.72\!\times\!10^{-21}$\\
\bottomrule
\end{tabular}
\end{table}

\subsubsection{Adaptive Selection from a Fixed Session Pool}

Let $n_s$ be the number of sessions available before verification feedback and define the complete ordered-pair family $\mathcal C=\{Q_{ab}:1\le a,b\le n_s,\ a\ne b\}$, of size $M=n_s(n_s-1)$. Each candidate has fixed prompt and metadata. The verifier is stateless and uses a fixed numerical verification map. The pool includes every session from which the adversary may construct a swap, not just a shortlist selected after inspecting content.

\begin{lemma}[Adaptive search over a finite pool]
\label{lem:pool_search}
If each fixed pair in $\mathcal C$ satisfies Eq.~\eqref{eq:baseline_uniform}, any selection strategy restricted to $\mathcal C$ satisfies
\begin{equation}
 \Pr[\mathrm{Forge}]\le\min\{1,M\varepsilon_N(0)\}.
 \label{eq:pure_pool_bound}
\end{equation}
The strategy may inspect all generated content, rank similarities, and adapt its query order to previous binary responses.
\end{lemma}
\begin{proof}
Let $E_{ab}=\{\mathcal D(Q_{ab})=1\}$. Success requires at least one accepted candidate in the fixed family, so
\[
 \{\mathrm{Forge}\}\subseteq\bigcup_{a\ne b}E_{ab},\qquad
 \Pr[\mathrm{Forge}]\le\sum_{a\ne b}\Pr[E_{ab}]
 \le M\varepsilon_N(0).
\]
The union bound uses only the individual event probabilities. Events may share target keys, video tracks, audio tracks, or decoding errors. Content-dependent ranking and feedback-dependent ordering change which candidate is submitted, but not the existence event bounded above.
\end{proof}

For a fixed target video, the candidate family has $n_s-1$ donor audios, giving $(n_s-1)\varepsilon_N(0)$. At $n_s=1{,}000$, $N=128$, and $\tau_{bind}=0.8$, the all-pairs and fixed-target bounds are approximately $9.66\times10^{-7}$ and $9.66\times10^{-10}$, respectively. These bounds account for offline selection over the entire allowed family, even if few candidates are ultimately sent to the API.

\subsubsection{Binary Feedback and the All-Reject Path}

A complementary reduction isolates the benefit of verification feedback. Let $K$ denote the fixed verification secrets and $\mathcal D_K(Q)\in\{0,1\}$ a stateless binary verifier. Let $V$ be the adversary's side information, including model weights, generated outputs, prompts, and public indices. All queries in the following experiment are admissible attack attempts: an Accept is already a false acceptance and terminates the first-success game. The interface provides no score, error category, or other feedback. The budget $q$ counts every tested attack candidate, including any final candidate checked by the experiment.

Define $\varepsilon_{single}$ as the supremum of success probabilities over all allowed candidate constructors with the same side information, computational resources, and generation access, but no verification feedback. Generation access may be simulated with the same oracle available in the real experiment. This definition includes content-dependent offline candidate selection.

\begin{theorem}[Binary-feedback first-success bound]
\label{thm:reject_path}
An adversary making at most $q$ attack-only verification queries satisfies
\begin{equation}
 \Pr[\mathrm{Forge}]\le\min\{1,q\varepsilon_{single}\}.
 \label{eq:reject_path_bound}
\end{equation}
No independence of queries, or bound on their history-conditioned success probabilities, is required.
\end{theorem}
\begin{proof}
Fix the adversary's random tape and couple any generation-oracle randomness between executions. Simulate a counterfactual interaction that answers every verification query with Reject. This generates an all-reject sequence
\[
 Q_t^\perp=\mathcal A_t(V,Q_1^\perp,0,\ldots,Q_{t-1}^\perp,0),
 \qquad t=1,\ldots,q.
\]
Here $\mathcal A_t$ is the next-query rule; random tapes and generation-oracle interactions are implicit. Suppose the real execution first succeeds on query $t$. All preceding responses are Reject, so its prefix agrees with the counterfactual execution, and $Q_t=Q_t^\perp$. Consequently,
\[
 \{\mathrm{Forge}\}\subseteq
 \bigcup_{t=1}^q\{\mathcal D_K(Q_t^\perp)=1\}.
\]
For every fixed $t$, a no-verification-feedback constructor can simulate the adversary with $t-1$ dummy rejections, using the same permitted generation access, and output $Q_t^\perp$. Its success probability is at most $\varepsilon_{single}$. Taking probabilities over the secret environment, public view, and coupled random tapes, then applying the union bound, yields
\[
 \Pr[\mathrm{Forge}]\le\sum_{t=1}^q
 \Pr[\mathcal D_K(Q_t^\perp)=1]\le q\varepsilon_{single}.
\]
The prefix argument holds even when the candidates are highly correlated or each is a modification of the previous one.
\end{proof}

Rejections can exclude possibilities: guessing a fresh $N$-bit password after one failed distinct guess has conditional success probability $1/(2^N-1)$, exceeding $2^{-N}$. Nevertheless, $q$ distinct guesses succeed with probability $q/2^N$ for $q\le2^N$. The reduction captures this distinction. The single-query quantity must bound every allowed no-feedback constructor, rather than only an unselected random swap. A $q\varepsilon_N(0)$ specialization follows if such constructors satisfy the stronger $\varepsilon_N(0)$ bound. With only the per-fixed-pair model, \cref{lem:pool_search} supplies the pool-size guarantee directly. Genuine-positive diagnostic queries fall outside this attack-only first-success experiment.

\subsubsection{Adaptive Local Perturbations of a Swap}

Fix a base swap $Q_0=(\mathbf x_v^{(A)},\mathbf x_a^{(B)})$, $A\ne B$, and candidates $Q_\delta=(\mathbf x_v^{(A)},\mathbf x_a^{(B)}+\delta)$. Let $\Delta_{\epsilon_{atk}}$ be a prescribed perturbation set, for example $\|\delta\|_p\le\epsilon_{atk}$ with any donor-content constraints. The budget is measured relative to the original donor audio, so successive small updates remain subject to one total budget. Prompts, request metadata, and numerical verification settings are fixed.

Let $U_0$ and $U_\delta$ be the recovered binding vectors. Define a good event $\mathcal G_b$ on which every $\delta\in\Delta_{\epsilon_{atk}}$ preserves the selected target record and satisfies $d_H(U_\delta,U_0)\le b$, where $b$ is a fixed integer. Suppose $\Pr[\mathcal G_b^c]\le\beta$. The server's ideal-grid reconstruction keeps the target digest $h$ fixed throughout this event.

\begin{lemma}[Adaptive neighborhood bound]
\label{lem:local_binding}
Under the baseline law for $U_0$ and the stability event above,
\begin{equation}
 \Pr[\exists\delta\in\Delta_{\epsilon_{atk}}:\mathcal D(Q_\delta)=1]
 \le\min\{1,\varepsilon_N(b)+\beta\}.
 \label{eq:local_binding_bound}
\end{equation}
\end{lemma}
\begin{proof}
On $\mathcal G_b$, an accepted perturbation has $d_H(U_\delta,h)\le r_\tau$. The triangle inequality gives
\[
 d_H(U_0,h)\le d_H(U_0,U_\delta)+d_H(U_\delta,h)\le b+r_\tau.
\]
Thus existence of an accepted perturbation implies either $d_H(U_0,h)\le r_\tau+b$ or failure of $\mathcal G_b$. The original uniform law for $U_0$ gives the enlarged-ball probability $\varepsilon_N(b)$; the failure event contributes at most $\beta$. This union argument does not condition the baseline distribution on stability, and requires no independence between those events.
\end{proof}

The bound covers every candidate in the allowed neighborhood at once, so adaptive query count does not enter. Perturbed vectors may be biased and mutually dependent. When $\tau_{bind}-b/N>1/2$, its statistical term also satisfies
\begin{equation}
 \varepsilon_N(b)\le
 \exp[-2N(\tau_{bind}-1/2-b/N)^2].
 \label{eq:local_exponential_bound}
\end{equation}
To see this, $d_H(U_0,h)\le r_\tau+b$ requires at least $\kappa-b>N\tau_{bind}-b$ baseline matches; apply Hoeffding above the mean $N/2$. Thus changing at most $b$ decoded bits reduces the effective threshold by $b/N$ in this exponential bound.

\begin{table}[ht]
\centering\small
\caption{\textbf{Analytical Neighborhood Bounds.} Prescribed decoded-bit budgets with $N=128$, $\tau_{bind}=0.8$, and $\beta=0$. These are parameter illustrations of Eq.~\eqref{eq:local_binding_bound}.}
\label{tab:local_binding_bounds}
\begin{tabular}{@{}ccc@{}}
\toprule
Bit-change budget $b$ & Radius $r_\tau+b$ & $\varepsilon_{128}(b)$\\
\midrule
0 & 25 & $9.67\times10^{-13}$\\
4 & 29 & $1.92\times10^{-10}$\\
8 & 33 & $1.91\times10^{-8}$\\
16 & 41 & $2.94\times10^{-5}$\\
\bottomrule
\end{tabular}
\end{table}

\subsubsection{From Numerical Inversion to Decoded-Bit Stability}

Let $\mathcal F$ denote the complete deterministic numerical mapping from queried media to recovered joint latents, including preprocessing, latent encoding, and the implemented inversion solver. Suppose a certified bound gives
\begin{equation}
 \sup_{\delta\in\Delta_{\epsilon_{atk}}}
 \|\mathcal F(Q_\delta)-\mathcal F(Q_0)\|_\infty\le\eta.
 \label{eq:inversion_stability_certificate}
\end{equation}
A local operator-norm bound $L_{\mathcal F}$ for the actual map from input $\ell_p$ perturbations to output $\ell_\infty$ changes supplies $\eta=L_{\mathcal F}\epsilon_{atk}$. The map in this statement includes the implemented finite-step solver; the continuous ODE alone is not the numerical certificate.

Let $\widetilde z_0=\mathcal F(Q_0)$. Any coordinate satisfying $|\widetilde z_0[j]|>\eta$ retains its sign throughout the neighborhood. For logical binding bit $i$, let $G_i$ be its physical repetition-copy set under the fixed recovered record, and let $s_j\in\{-1,+1\}$ be each baseline sign after XOR derandomization. Define the vote margin and the number of potentially changing copies as
\begin{equation}
 \gamma_i=\left|\sum_{j\in G_i}s_j\right|,\qquad
 e_i(\epsilon_{atk})=\#\{j\in G_i:|\widetilde z_0[j]|\le\eta\}.
 \label{eq:vote_margin_certificate}
\end{equation}

\begin{lemma}[Majority-margin certificate]
\label{lem:margin_certificate}
If Eq.~\eqref{eq:inversion_stability_certificate} holds and the target record is fixed, then
\begin{equation}
 b_{cert}(\epsilon_{atk})=\#\{i:\gamma_i\le2e_i(\epsilon_{atk})\}
 \label{eq:bit_change_certificate}
\end{equation}
upper-bounds the number of decoded binding bits that can change under any allowed perturbation.
\end{lemma}
\begin{proof}
Only copies with $|\widetilde z_0[j]|\le\eta$ can change sign. Each flipped copy changes its signed vote by at most 2, so bit $i$'s vote sum changes by at most $2e_i$. If $\gamma_i>2e_i$, the sum cannot cross or reach zero, and its majority decision stays fixed. Counting all remaining groups gives Eq.~\eqref{eq:bit_change_certificate}. Groups at equality, including ties under the decoder's fixed tie rule, are conservatively counted as potentially changing. Fixed permutation only assigns copies to groups, and XOR only reverses their signs, so neither changes the counting argument.
\end{proof}

Apply the same margin test to the public, unmasked index-copy groups $\mathcal J_I(\ell)$ for $\ell=1,\ldots,L_I$. Certifying every index bit preserves the lookup result; fixed prompt and metadata then preserve the session key, permutation, and masks used in the binding certificate. The operational chain is therefore
\[
 \text{input budget}\ \Longrightarrow\ \eta\
 \Longrightarrow\ \text{copy-sign margins}\
 \Longrightarrow\ (\text{stable index},b_{cert})\
 \Longrightarrow\ \text{binding bound}.
\]
For probabilistic reporting, choose a fixed budget $b$ and bound the probability that the inversion certificate, index certificate, or inequality $b_{cert}\le b$ fails. This yields $\beta$ for a pair or $\beta_g$ jointly for a pool. A sample-dependent certificate is used through this failure event, rather than by conditioning the uniform baseline law on its observed value. No independence between certificate failures and initial binding vectors is needed.

\subsubsection{Composition over Session Pools and Perturbation Neighborhoods}

Let $\mathcal G_b^{pool}$ be the event that every base pair in $\mathcal C$ and every permitted perturbation preserves its record and changes at most $b$ binding bits. Set $\Pr[(\mathcal G_b^{pool})^c]\le\beta_g$. This event controls the entire allowed search family, including switches between targets, donor audios, and local modifications.

\begin{proof}[Proof of \cref{thm:security}]
On $\mathcal G_b^{pool}$, acceptance anywhere in the search family implies $d_H(U_{ab},h_{ab})\le r_\tau+b$ for at least one fixed base pair. Hence
\[
 \{\mathrm{Forge}\}\subseteq
 \bigcup_{a\ne b}\{d_H(U_{ab},h_{ab})\le r_\tau+b\}
 \ \cup\ (\mathcal G_b^{pool})^c.
\]
Apply the baseline uniform law separately to each pair and union-bound the $M$ enlarged-ball events. Charge the global failure event once. This gives $M\varepsilon_N(b)+\beta_g$, capped at 1. Adaptive decisions may exploit content and any earlier accept/reject feedback because the bound concerns existence of a successful candidate anywhere in the allowed family.
\end{proof}

If only a per-pair failure bound $\beta$ is available, the corresponding bound is
\begin{equation}
 \Pr[\mathrm{Forge}]\le\min\{1,M[\varepsilon_N(b)+\beta]\}.
 \label{eq:pairwise_stability_pool}
\end{equation}
For pure swaps, $b=0$ and $\beta_g=0$, recovering \cref{lem:pool_search}. The finite-pool and neighborhood statements are statistical results in the ideal baseline model; \cref{thm:reject_path} separately converts any established no-feedback success bound into a binary-query bound. Any real-to-ideal cryptographic reduction is accounted for through its distinguishing advantage: if it bounds the relevant adversary-success probability gap by $\nu$, add $\nu$ to the ideal bound and cap at 1. This keeps the statistical search factor and cryptographic reduction loss separate.

\section{Inversion Dynamics and Solver Analysis}
\label{app:solver}

\subsection{ODE Solver Discretization and Truncation Drift}
In this section only, $N$ counts solver steps, not binding bits; $\epsilon=\delta_t$ is the final integration time. The forward ODE for Rectified Flow~\citep{liu2022flow} models define the velocity field:
\begin{equation}
    \frac{d\mathbf{z}_t}{dt} = v_\theta(\mathbf{z}_t, t, c),
\end{equation}
where $\mathbf{z}_0 \sim \mathcal{N}(0, I)$ is the initial noise and $\mathbf{z}_1$ is the generated sample. For inversion, we solve the ODE backwards from $t=1$. We use the Euler method~\citep{karras2022elucidating} with a uniform time schedule $\{t_N = 1, t_{N-1}, \dots, t_1 = \epsilon\}$ where $N$ is the number of inversion steps:
\begin{equation}
    \mathbf{z}_{t_{i-1}} = \mathbf{z}_{t_i} + (t_{i-1} - t_i) \cdot v_\theta(\mathbf{z}_{t_i}, t_i, c).
\end{equation}

\textbf{The $\delta_t$ Truncation Gap.}
In theoretical implementations, integrating exactly to $t=0$ yields perfect inversion. However, to ensure generation quality and avoid the numerical singularity near pure noise ($t=0$) where the velocity field $v_\theta$ may become unbounded, state-of-the-art joint audio-visual generation models (MOVA and LTX-2) enforce a strict minimum timestep $t=\delta_t > 0$ in their sampling pipelines. Consequently, the inversion process must be truncated to maintain stability. This discretization gap introduces a deterministic drift, preventing the raw Bit Accuracy from reaching a perfect $1.0$.

\subsection{Inversion Fidelity: DDIM vs. Rectified Flow}
Prior works utilizing DDIM~\citep{song2021ddim} solve a related ODE in the noise-prediction parameterization:
\begin{equation}
    \mathbf{z}_{t_{i-1}} = \sqrt{\alpha_{t_{i-1}}} \left( \frac{\mathbf{z}_{t_i} - \sqrt{1-\alpha_{t_i}} \cdot \epsilon_\theta(\mathbf{z}_{t_i}, t_i)}{\sqrt{\alpha_{t_i}}} \right) + \sqrt{1-\alpha_{t_{i-1}}} \cdot \epsilon_\theta(\mathbf{z}_{t_i}, t_i).
\end{equation}

To systematically understand the impact of the solver paradigm versus the aforementioned $\delta_t$ truncation, we compare inversion fidelity across three distinct experimental settings:

\begin{itemize}
    \item \textbf{DDIM-based Inversion (VideoShield~\citep{hu2025videoshield}):} Evaluated on ModelScope~\citep{wang2023modelscope} (T2V) and Stable-Video-Diffusion~\citep{blattmann2023stable} (I2V), integrating fully to the noise boundary.
    \item \textbf{Ideal Rectified Flow Inversion:} Evaluated using the official Rectified Flow codebase on CIFAR-10~\citep{krizhevsky2009learning}, utilizing an Euler ODE solver that integrates fully to $t=0$.
    \item \textbf{mAVE (Ours) on MOVA-720p:} Operating under the enforced $\delta_t$ truncation inherently required by native joint generation pipelines.
\end{itemize}

\begin{table}[h]
    \centering
    \caption{\textbf{Inversion Accuracy vs.\@ Number of Steps across Frameworks.} Both DDIM and ideal Rectified Flow achieve near-perfect inversion ($\sim 1.000$) given sufficient steps. In our mAVE pipeline, performance is marginally bottlenecked by the necessary $\delta_t$ truncation rather than the underlying solver paradigm, yet it remarkably maintains high robustness even down to $5$ steps.}
    \label{tab:inversion_comparison}
    \begin{tabular}{l|c|ccccc}
        \toprule
        \multirow{2}{*}{\textbf{Framework / Model}} & \multirow{2}{*}{\textbf{Task}} & \multicolumn{5}{c}{\textbf{Inversion Steps ($N$)}} \\
        & & 50 & 25 & 10 & 5 & 1 \\
        \midrule
        \multicolumn{7}{l}{\textit{Baseline Paradigms (Full Integration to $t=0$)}} \\
        \midrule
        DDIM (ModelScope) & T2V & 1.000 & 1.000 & 1.000 & - & - \\
        DDIM (SVD) & I2V & 0.999 & 0.999 & 0.999 & - & - \\
        Rectified Flow (CIFAR-10) & Uncond & 1.000 & 1.000 & 1.000 & - & - \\
        \midrule
        \multicolumn{7}{l}{\textit{mAVE with $\delta_t$ Truncation (MOVA-720p)}} \\
        \midrule
        mAVE Video BA & \multirow{2}{*}{TI2AV} & 0.967 & 0.961 & 0.955 & 0.958 & 0.857 \\
        mAVE Audio BA & & 0.913 & 0.892 & 0.901 & 0.865 & 0.780 \\
        \bottomrule
    \end{tabular}
\end{table}

As demonstrated in Table~\ref{tab:inversion_comparison}, both DDIM and Flow Matching naturally provide excellent inversion guarantees in practice when integrated fully. For mAVE, the slight performance gap is fundamentally a product of the $\delta_t$ truncation mechanism, rather than an inherent weakness of the Euler solver. 

Crucially, despite this truncation drift, the straight-line transport property of Rectified Flow allows mAVE to remain highly resilient under extremely constrained computational budgets. The accuracy is solidly maintained when scaled down to just $5$ inversion steps, with significant degradation only occurring at an extreme $1$-step inversion. This property is highly advantageous for efficient, real-world deployment.

\section{Decoding Method Comparison: Median versus Zero Thresholding}
\label{app:decoding}

After inversion, the recovered latent $\tilde{\mathbf{z}}_0$ exhibits minor drift due to the discretization gap. We investigated two decoding strategies:

\textbf{Zero-Thresholding (Standard).} The simplest approach: each bit is decoded as:
\begin{equation}
    \tilde{b}[i] = \mathbb{1}[\tilde{\mathbf{z}}_0[i] > 0].
\end{equation}
This leverages the fundamental asymmetry of the half-Gaussian constraint: watermarked coordinates are guaranteed to have the correct sign in the noise-free case, and the redundancy from block-wise repetition absorbs moderate sign flips.

\textbf{Median Thresholding.} This approach estimates a per-block threshold by computing the median of all $k_{all}$ copies of each bit, then decodes relative to this median:
\begin{equation}
    \mu_j = \text{median}\{\tilde{\mathbf{z}}_0[i] : i \in \text{Block}_j\}, \quad \tilde{b}[j] = \mathbb{1}[\mu_j > 0].
\end{equation}
The motivation is to correct for a potential global bias in the recovered latents. However, since Rectified Flow inversion introduces very little systematic bias (the drift is primarily random noise), the median and zero thresholds produce nearly identical results.

\begin{table}[h]
    \centering
    \caption{\textbf{Decoding Strategy Comparison.} Bit Accuracy on LTX-2 T2AV with 25 inversion steps. The negligible difference confirms that zero-thresholding is sufficient.}
    \label{tab:decoding_comparison}
    \begin{tabular}{l|cc}
        \toprule
        \textbf{Decoding Method} & \textbf{Video BA} & \textbf{Audio BA} \\
        \midrule
        Zero-Thresholding & 0.936 & 0.915 \\
        Median Thresholding & 0.940 & 0.913 \\
        \bottomrule
    \end{tabular}
\end{table}

As shown in Table~\ref{tab:decoding_comparison}, the performance difference is marginal and inconsistent across modalities. We therefore adopt the simpler Zero-Thresholding in all experiments to avoid introducing an unnecessary hyperparameter.
\section{The Performance Analysis of Joint Inversion}
\label{app:joint_vs_separate}

A key architectural advantage of mAVE is that both embedding and extraction operate through the \textit{same} joint denoising / inversion pass used by the underlying audio-visual model. In this section we empirically verify that this single joint inversion recovers watermark bits as accurately as performing two dedicated, modality-specific inversions, and we discuss the resulting efficiency gain.

\subsection{Setup}
We compare two extraction protocols on LTX-2 T2AV (25 inversion steps, zero-thresholding) using the same four NVIDIA RTX 5090 GPUs and timing protocol. Joint inversion has reported peak/average VRAM usage of 32/22\,GB. The timing comparison is between one joint inversion and two separate modality inversions:
\begin{itemize}
    \item \textbf{Joint Inversion (default).} A single backward ODE pass is performed through the bimodal model, simultaneously recovering both the video latent $\tilde{\mathbf{z}}_v$ and the audio latent $\tilde{\mathbf{z}}_a$.
    \item \textbf{Separate Inversion.} Two independent backward ODE passes are performed: one using only the video branch to recover $\tilde{\mathbf{z}}_v$, and one using only the audio branch to recover $\tilde{\mathbf{z}}_a$. Each pass uses the same 25 inversion steps and solver configuration.
\end{itemize}

\subsection{Results}

\begin{table}[h]
    \centering
    \caption{\textbf{Joint vs.\ Separate Inversion (LTX-2 T2AV).} Bit Accuracy and wall-clock time comparison. Joint inversion matches the extraction quality of two separate passes while halving the computational cost.}
    \label{tab:joint_vs_separate}
    \begin{tabular}{l|cc|c}
        \toprule
        \textbf{Extraction Protocol} & \textbf{Video BA} & \textbf{Audio BA} & \textbf{Time (s/item)} \\
        \midrule
        Joint Inversion (ours) & 0.936 & 0.915 & 2.07 \\
        Separate Inversion & 0.936 & 0.915 & 4.16 \\
        \bottomrule
    \end{tabular}
\end{table}

As shown in Table~\ref{tab:joint_vs_separate}, the two protocols yield virtually identical Bit Accuracy for both modalities, confirming that the joint denoising trajectory preserves per-modality fidelity during inversion. This is expected: in Rectified-Flow-based joint models, the audio and video latents share the same ODE velocity field but occupy orthogonal sub-spaces of the joint latent, so inverting them together introduces no additional cross-modal interference compared with inverting them separately.

\subsection{Efficiency Implications}
The practical consequence is a significant reduction in detection cost relative to baseline combinations. Table~\ref{tab:efficiency_comparison} summarises the pipeline-level comparison.

\begin{table}[h]
    \centering
    \caption{\textbf{Extraction Pipeline Comparison.} mAVE uses one joint inversion; baseline combinations additionally invoke an audio extractor. Pass counts alone do not establish wall-clock speedups across different extractors.}
    \label{tab:efficiency_comparison}
    \resizebox{\textwidth}{!}{%
    \begin{tabular}{l|ccc|c}
        \toprule
        \textbf{Method} & \textbf{Video Extraction} & \textbf{Audio Extraction} & \textbf{Binding Check} & \textbf{Total Passes} \\
        \midrule
        VideoShield + AudioSeal & ODE inversion ($N$ steps) & AudioSeal encoder (forward) & --- & 2 \\
        VideoShield + WavMark & ODE inversion ($N$ steps) & WavMark encoder (forward) & --- & 2 \\
        \textbf{mAVE (ours)} & \multicolumn{2}{c|}{Joint ODE inversion ($N$ steps)} & $\tau_{bind}$ & \textbf{1} \\
        \bottomrule
    \end{tabular}%
    }
\end{table}

The measured latency decreases from 4.16 to 2.07\,s/item, a $2.01\times$ speedup (50.2\% reduction), while video/audio BA remains 0.936/0.915. This supports replacing two separate inversions with one joint pass on the measured setup. It does not establish the same speedup relative to VideoShield plus AudioSeal or WavMark, whose neural-extractor costs were not measured here. Likewise, the reported hardware and VRAM figures characterize this deployment rather than multi-GPU scaling or single-GPU feasibility.
\section{Complete Cryptographic Key Derivation}
\label{app:key_derivation}

This section provides the end-to-end derivation of the session-specific randomization, from payload to continuous latent initialization.

The consolidated symbol glossary is in \cref{app:notation}.

\subsection{Payload Structure: Index and Secret}
The overall watermark payload comprises two functionally distinct segments:
\begin{itemize}
    \item \textbf{Plaintext Index} $I \in \{0,1\}^{L_I}$ ($L_I = 32$ bits) \textit{(embedded)}: A non-confidential session identifier repeated at fixed public video coordinates $\mathcal J_I$. These coordinates bypass both secret permutation and ChaCha20 masking. Their inverse-sampled signs permit majority-vote index recovery without a key, after prompt-conditioned inversion. No digest, template, or secret payload shares these coordinates.
    \item \textbf{Secret Payload} $m \in \{0,1\}^{L_m}$ \textit{(non-embedded)}: A high-entropy secret stored exclusively on the server-side database, indexed by $I$. The payload $m$ is \textit{never} embedded in the latent; it participates only in key derivation.
\end{itemize}

\subsection{Session Key Computation}
Given the secret payload $m$ and the generation prompt $P$, the session key is computed as:
\begin{equation}
    K_{sess} = \text{SHA-256}\bigl(\text{Prefix}_{\lambda_{key}}(\text{SHA-256}(m)) \;\|\; \text{SHA-256}(\text{Norm}(P))\bigr),
\end{equation}
where $\text{Prefix}_{\lambda_{key}}(\cdot)$ retains a fixed key-security prefix; $\lambda_{key}$ is independent of the binding length $N$, $\|$ denotes concatenation, and $\text{Norm}(\cdot)$ applies lowercase conversion and whitespace trimming to the prompt string. The prompt binds the key derivation to the generation condition; it is not a password or a source of secret entropy. The adversary may know it. Even when prompts coincide after normalization, fresh independent session secrets produce distinct session keys with overwhelming probability. Changing the prompt changes derived material, but does not by itself prove resistance to general copy attacks.

At generation time, the provider stores the original prompt $P$, public index $I$, and secret $m$ in the session record. Before inversion, the original $P$ is either supplied by the caller or retrieved by the provider using the generation-request metadata available to the verification service. This service-side request reference is distinct from the embedded index $I$. The verifier conditions the joint ODE inversion on $P$, recovers $\hat I$ from the resulting video latent, and uses $\hat I$ to retrieve $m$ and derive $K_{sess}$ with $E_P=\operatorname{Norm}(P)$. Thus prompt acquisition precedes inversion and watermark-index extraction; $P$ is not inferred from the received content. Key secrecy rests on the server-held $m$, not on hiding $P$. Normalization affects key derivation only; the generator is conditioned on the original prompt. Metadata and $I$ locate records rather than authenticate content. The protocol assumes trusted session records; acceptance is decided by recovered payload and binding checks.

\paragraph{Purpose-separated keys.}
Let $\operatorname{KDF}$ denote the purpose-separated derivation and let $\mathrm{label}_x$ be distinct fixed application labels. We use
\[
 K_x=\operatorname{KDF}(K_{sess},\mathrm{label}_x),\qquad
 x\in\{grid,mask,perm\}.
\]
Grid construction uses the grid-purpose subkey together with the shared seed defined below; block-shuffle randomness and per-copy masking use their own domains. Purpose-label separation follows the standard KDF design principle~\citep{krawczyk2010hkdf}. The joint randomization premise and its reduction are specified in \cref{app:losslessness}. The fixed $\lambda_{key}$ controls key derivation throughout the binding-length ablation; changing $N$ changes digest coverage, not key-prefix length.

\subsection{Shared Seed Extraction}
To ensure correlated base entropy across the audio and video modalities, we derive a shared seed from the secret payload $m$:
\begin{equation}
    \text{seed}_{shared} = \text{SHA-256}(m)[:8],
\end{equation}
\ie, the first 8 bytes (64 bits) of the hash of the secret payload $m$. This shared seed initializes both the audio and video pseudo-random streams, ensuring that their bit grids are structurally aligned for the binding verification.

\subsection{ChaCha20 Stream Generation}
The ChaCha20 stream cipher~\citep{bernstein2008chacha} is initialized with:
\begin{itemize}
    \item \textbf{Key:} Mask-purpose subkey $K_{mask}$ (256-bit).
    \item \textbf{Nonce:} A fixed zero nonce (96-bit all-zeros). This fixed nonce requires a fresh masking key for each generation request. Reusing a prompt does not authorize reusing the session secret or a key/nonce pair; video/audio masks occupy disjoint stream slices.
    \item \textbf{Counter:} Starting from 0.
\end{itemize}
ChaCha20~\citep{nir2018chacha} generates the masking byte stream $\mathcal S$ from $K_{mask}$, with disjoint video/audio slices and a distinct pseudorandom mask bit for every physical copy. Separate randomness derived from $K_{perm}$ drives the uniform shuffle of eligible blocks; $\pi$ is its induced coordinate map. Public index copies remain fixed. At $j\in\mathcal J_I$, the sign bit is the unmasked index bit; at $j\in\mathcal J_R$, it is the repeated, permuted payload bit XOR its per-copy mask.

\subsection{Public Index, Encoding, and Verification}
\label{app:index_protocol}
\textbf{Public layout.} Let $r_I$ be the number of repetitions per index bit and $\mathcal J_I(b)$ the fixed set of $r_I$ video coordinates assigned to index bit $b\in\{1,\ldots,L_I\}$. The sets are disjoint and known before generation; $\mathcal J_I=\bigcup_b\mathcal J_I(b)$. They do not depend on $I$, the prompt, or a secret. All keyed permutations and mask slices act only on $\mathcal J_R$. Other unused coordinates, if any, are sampled independently from the standard Gaussian. The public index provides indexed record retrieval before secret-dependent decoding, avoiding a search over all stored session keys. This routing function is separate from payload authentication: knowing an address does not establish a valid binding.

\textbf{Coordinate accounting.} Let $d_v$ and $d_a$ be the numbers of scalar initial-latent coordinates in video and audio. Public index copies are disjoint, so $|\mathcal J_I|=L_I r_I$. Define the video and joint public-coordinate fractions as
\[
 \eta_v=\frac{L_I r_I}{d_v},\qquad
 \eta_{AV}=\frac{L_I r_I}{d_v+d_a}.
\]
The complementary joint fraction is $1-\eta_{AV}$, and audio contains no public index copies. These counts concern physical latent coordinates, not the ratio of logical index bits to payload bits. The payload repetition product is $k_{all}=k_ck_tk_hk_w=48$ for the reported $(3,1,4,4)$ setting; $r_I$ is the separate index repetition count. Repeat fixed logical payload fields first, uniformly permute the resulting blocks within the eligible region of each modality, then apply a distinct pseudorandom mask bit to every physical copy. Even copies of the same payload bit do not reuse a mask bit. Audio reshaping and any additional replication must be accounted for in $d_a$ and the corresponding decoding map.

\textbf{Binding positions and digest bits.} The logical fields for base bits, templates, and binding bits have fixed allocation. Audio binding positions $\mathcal I_{bind}$ are disjoint from template positions $\mathcal I_{time}$; their physical blocks are relocated by the secret uniform block permutation. Slot selection is thus fixed in logical space and secret in physical space. The fixed digest selector $\phi$ uses 128 different SHA-256 bits by default, or $N$ distinct bits in the binding-length ablation ($N\le256$); repeated physical copies provide decoding redundancy and do not increase this number. This separates logical payload allocation, physical coordinate mapping, and the number of comparisons used in the probability bound.

\textbf{Encoding pseudocode.}
\begin{enumerate}
\item Allocate an unused index $I$ and a fresh independent secret $m$; store $(I,m,P)$ in the provider's session record. Derive $K_{sess}$ as in Eq.~\eqref{eq: key}, then the grid, mask, and permutation subkeys.
\item Construct $B_v$ with its template, index, and base bits. Compute its digest and construct $B_a$ with the designated binding bits. Repeat the grids to obtain $B^{diff}$.
\item Keep index copies at $\mathcal J_I(b)$. Reconstruct the keyed block permutation and disjoint video/audio per-copy mask slices on $\mathcal J_R$. Set $M_{rand}[j]=I_b$ for $j\in\mathcal J_I(b)$ and $M_{rand}[j]=B^{diff}[j]\oplus\mathcal K[j]$ on the permuted complement.
\item Draw independent uniform variates and apply Eq.~\eqref{eq:sampling_formula}; the practical version uses the clipping specified below. Generate the AV pair with the original prompt $P$.
\end{enumerate}

\textbf{Verification pseudocode.}
\begin{enumerate}
\item Obtain the original prompt from the caller or the provider's request metadata. Jointly invert the received AV pair once to recover approximate initial latents. No watermark key is needed for this model inversion.
\item Read signs at the public positions and decode each index bit as $\hat I_b=\mathbb 1[\sum_{j\in\mathcal J_I(b)}\mathbb 1[\tilde z_j>0]>r_I/2]$; ties decode to zero. No keyed permutation or XOR is applied at this stage.
\item Query the session database with $\hat I$ to obtain $m$, then derive $K_{sess}$ and its purpose subkeys from $m$ and the normalized original prompt. This is the first step requiring the session secret.
\item Reconstruct the block permutation and per-copy mask slices on $\mathcal J_R$; undo masking and permutation, then majority-decode the payload. Rebuild the ideal grids and video digest from the retrieved record.
\item Apply the video/audio BA gates and the binding gate in Eq.~\eqref{eq:and_gate}. Successful index recovery alone does not authenticate either modality or their pairing.
\end{enumerate}

\textbf{Routing, authentication, and fidelity.} Public repetition enables direct record lookup; secret randomization protects the authentication payload. Theorem~\ref{thm:lossless} covers the complete numerical sampler and generated output against a reference with the same routing header. \cref{app:losslessness} additionally gives the exact untruncated joint law, the index entropy cost, and the full routing/clipping transport budget. Tables~\ref{tab:fidelity} and~\ref{tab:index_recovery} report complete-system quality and index recovery.

\textbf{Exact index recovery.} Let $n_I$ be the number of evaluated samples. The exact-index recovery rate is $R_I=n_I^{-1}\sum_{k=1}^{n_I}\mathbb 1[\hat I^{(k)}=I^{(k)}]$, where equality requires all $L_I=32$ bits to match. The denominator includes every attempted sample, not only successful lookups. Index BA and payload BA cannot substitute for this rate. On the clean sets in Table~\ref{tab:index_recovery}, exact index recovery is 100\% for each task. This separately measures record recovery before payload detection and binding acceptance.
\begin{table}[ht]
\centering\small
\caption{Exact 32-bit index recovery on clean samples. Every index bit must match; all attempted samples are included. These rates do not describe attacked samples.}
\label{tab:index_recovery}
\begin{tabular}{@{}lc@{}}
\toprule
Model / task & Exact index recovery (\%)\\
\midrule
LTX-2 / T2AV & 100\\
LTX-2 / I2AV & 100\\
MOVA-720p / TI2AV & 100\\
\bottomrule
\end{tabular}
\end{table}

\subsection{Numerical Inverse Sampling}
For each watermark coordinate $i\in\mathcal I_w$ with assigned sign bit $y_i$, draw an independent $u_i\sim\mathcal U(0.5,1)$ and apply the actual sample-then-clamp rule
\begin{equation}
 u_i^\delta=\min\{1-\delta,\max\{0.5+\delta,u_i\}\},\qquad
 z_i=\Phi^{-1}\!\left(y_i u_i^\delta+(1-y_i)(1-u_i^\delta)\right).
 \label{eq:implemented_clamp}
\end{equation}
We use $\delta=0.001$, giving $|z_i|\in[\Phi^{-1}(0.501),\Phi^{-1}(0.999)]\approx[0.00250663,3.09023231]$. Sampling precedes clamping, so the boundary values carry probability mass. Theorem~\ref{thm:lossless} includes this numerical operator, and \cref{app:losslessness} quantifies its near-zero and tail contributions to the prior budget. Non-watermark coordinates $j\notin\mathcal I_w$ are sampled from the standard Gaussian.
\section{Formal Metric Definitions}
\label{app:metrics}

\subsection{Bit Decision Rule (Majority Voting)}
Here $m_j$ denotes a generic payload bit, not the server secret $m$; $k_{all}=k_ck_tk_hk_w$ is its repetition count, and $\pi(j,l)$ locates copy $l$ of bit $j$. Hats denote decoded estimates and tildes denote recovered values.
The block-wise repetition scheme embeds each payload bit $m_j$ into $k_{all}$ copies across the latent tensor. After inversion, the recovered copies are $\{\tilde{b}_{j,1}, \tilde{b}_{j,2}, \dots, \tilde{b}_{j,k_{all}}\}$ where $\tilde{b}_{j,l} = \mathbb{1}[\tilde{\mathbf{z}}_0[\pi(j,l)] > 0]$. The final decision for bit $j$ is:
\begin{equation}
    \hat{m}_j = \mathbb{1}\!\left[\sum_{l=1}^{k_{all}} \tilde{b}_{j,l} > \frac{k_{all}}{2}\right].
\end{equation}
That is, $\hat{m}_j = 1$ if the majority of the $k_{all}$ copies decode to 1, and $\hat{m}_j = 0$ otherwise.

\subsection{Bit Accuracy}
Given the ground-truth payload $m \in \{0,1\}^L$ and the recovered payload $\hat{m} \in \{0,1\}^L$ (after majority voting), the Bit Accuracy is defined as:
\begin{equation}
    \text{BA} = \frac{1}{L} \sum_{j=1}^{L} \mathbb{1}[\hat{m}_j = m_j].
\end{equation}
Bit Accuracy of 1.0 indicates perfect recovery, while 0.5 corresponds to random guessing. BA measures the fraction of correctly recovered bits, not the fraction of samples in which a watermark is detected. A sample need not recover every bit to pass the watermark-presence threshold. In our experiments, we report BA separately for the video modality ($\text{BA}_v$) and audio modality ($\text{BA}_a$).

\textbf{Sample-level watermark detection.} For each modality, detection rate is the percentage of authentic samples whose recovered BA exceeds the frozen threshold $\tau_{acc}$. Joint watermark detection requires both modality BA gates to pass, before applying the binding test in \cref{eq:and_gate}. Thus extraction BA, watermark detection rate, and final authentication TPR are distinct quantities. Tables~\ref{tab:extraction} and~\ref{tab:swap_defense} reserve separate columns for detection rates; these require per-sample threshold decisions and cannot be inferred from mean BA. For non-mAVE baselines, use each method's fixed watermark-presence detector. A reported 100\% detection rate would mean every evaluated authentic sample passed the relevant detector, not perfect bit recovery or a guarantee on unseen samples.

\textbf{Empirical false acceptance.} Let $a_-$ be accepted swapped pairs and $n_-$ all attempted swapped pairs for a given construction; the observed FPR is $100a_-/n_-$ percent. In Table~\ref{tab:swap_defense}, each method is evaluated on 1{,}000 independently sampled swaps. Each cross-prompt, same-prompt, and similarity-selected group in Table~\ref{tab:binding_controls} also contains 1{,}000 independent swap trials, one per target video. mAVE accepts none of the swapped pairs in each evaluation, giving $a_-/n_-=0/1{,}000$. For zero acceptances in $n_-$ independent trials, the one-sided upper confidence limit at level $1-\alpha$ is $1-\alpha^{1/n_-}$. With $\alpha=0.05$ and $n_-=1{,}000$, this is approximately $0.002991$, or $0.30\%$. The confidence limit applies separately to each test group. Authentic-pair counts separately determine TPR/FNR.

\subsection{Video Quality Metrics}
We adopt the five video quality metrics from VBench~\citep{huang2024vbench}. All scores are normalized to $[0,1]$; higher is better. Further details are available in the original VBench paper.

\textbf{Subject Consistency.}
Subject Consistency measures the temporal coherence of the main subject across frames by computing the DINO~\citep{caron2021emerging} feature similarity:
\begin{equation}
    S_{\text{subject}} = \frac{1}{T-1} \sum_{t=2}^{T} \frac{1}{2} \bigl( \langle d_1 \cdot d_t \rangle + \langle d_{t-1} \cdot d_t \rangle \bigr),
\end{equation}
where $d_i$ is the DINO image feature of the $i$-th frame (normalized to unit length) and $\langle \cdot \rangle$ denotes the dot product (cosine similarity).

\textbf{Background Consistency.}
Background Consistency evaluates the temporal consistency of background scenes by computing CLIP~\citep{radford2021learning} feature similarity across frames:
\begin{equation}
    S_{\text{background}} = \frac{1}{T-1} \sum_{t=2}^{T} \frac{1}{2} \bigl( \langle c_1 \cdot c_t \rangle + \langle c_{t-1} \cdot c_t \rangle \bigr),
\end{equation}
where $c_i$ represents the CLIP image feature of the $i$-th frame, normalized to unit length.

\textbf{Motion Smoothness.}
Motion Smoothness is evaluated by the frame-by-frame motion prior to video frame interpolation models~\citep{li2023amt}. Specifically, given a generated video consisting of frames $[f_0, f_1, f_2, f_3, f_4, \ldots, f_{2n-2}, f_{2n-1}, f_{2n}]$, the odd-numbered frames are manually dropped to obtain a lower-frame-rate video $[f_0, f_2, f_4, \ldots, f_{2n-2}, f_{2n}]$, and video frame interpolation is used to infer the dropped frames $[\hat{f}_1, \hat{f}_3, \ldots, \hat{f}_{2n-1}]$. The Mean Absolute Error (MAE) between the reconstructed frames and the original dropped frames is then computed and normalized into $[0,1]$, with a larger number implying smoother motion.

\textbf{Dynamic Degree.}
Dynamic Degree assesses the extent to which models tend to generate non-static videos. RAFT~\citep{teed2020raft} is used to estimate optical flow strengths between consecutive frames of a generated video. The average of the largest 5\% optical flows (capturing the movement of salient objects) is taken as the basis to determine whether the video is static. The final dynamic degree score is calculated by measuring the proportion of non-static videos generated by the model.

\textbf{Imaging Quality.}
Imaging Quality is measured by the MUSIQ~\citep{ke2021musiq} image quality predictor trained on the SPAQ dataset, which is capable of handling variable-sized aspect ratios and resolutions. The frame-wise score is linearly normalized to $[0,1]$ by dividing by 100, and the final score is then calculated by averaging the frame-wise scores across the entire video sequence.
\section{Local Temporal Perturbations}
\label{app:temporal}

This section evaluates local frame smoothing and adjacent-frame swaps using MarkDiffusion~\citep{pan2025markdiffusion}. The reported results concern these two operations; they do not establish robustness to arbitrary temporal transformations.

\subsection{Attack Definitions}
Let $\{f_t\}_{t=1}^{T}$ denote the original video sequence of $T$ frames.

\noindent\textbf{FrameAverage.}
A sliding-window smoothing attack that replaces each frame with the average of its $n$-neighbourhood ($n=3$):
\begin{equation}
    f'_t = \frac{1}{2n+1}\sum_{\tau = t-n}^{t+n} f_\tau,
\end{equation}
where boundary frames are handled via zero-padding. The output has the same length $T$ but with blurred pixel values that suppress high-frequency watermark signals.

\noindent\textbf{FrameSwap.}
A stochastic local permutation that, for each frame $f_t$ ($t \ge 2$), swaps it with its predecessor $f_{t-1}$ with probability $p=0.25$:
\begin{equation}
    (f'_{t-1},\, f'_t) =
    \begin{cases}
        (f_t,\, f_{t-1}), & \text{with prob.}\ p, \\
        (f_{t-1},\, f_t), & \text{otherwise}.
    \end{cases}
\end{equation}
The total frame count remains $T$, but the temporal ordering is locally disrupted.

\subsection{Experimental Setup}
We evaluate three configurations to isolate the effect of the watermarking method and the generation backbone:
\begin{enumerate}
    \item \textbf{VideoShield (ModelScope)}: The original VideoShield algorithm~\citep{hu2025videoshield} running on ModelScope~\citep{wang2023modelscope}.
    \item \textbf{VideoShield (LTX-2 T2AV)}: VideoShield ported to the same LTX-2 T2AV backbone used by mAVE.
    \item \textbf{mAVE (LTX-2 T2AV)}: Our full method on LTX-2 T2AV.
\end{enumerate}

\begin{table}[h]
    \centering
    \caption{\textbf{Temporal Attack Results (Video Bit Accuracy).} FrameAverage and FrameSwap evaluated on three configurations. Entries report video bit accuracy, not sample-level detection or pair-authentication rates.}
    \label{tab:temporal_attacks}
    \begin{tabular}{l|cc}
        \toprule
        \textbf{Configuration} & \textbf{FrameAvg} & \textbf{FrameSwap} \\
        \midrule
        VideoShield (ModelScope) & 0.995 & 0.732 \\
        VideoShield (LTX-2 T2AV) & 0.948 & 0.699 \\
        mAVE (LTX-2 T2AV) & 0.927 & 0.710 \\
        \bottomrule
    \end{tabular}
\end{table}

\subsection{Analysis}
Across the three configurations, FrameAverage retains higher BA than FrameSwap, consistent with the importance of temporal alignment for extraction.

The evaluated schemes recover payloads from indexed latent grids. For mAVE, $k_t=1$ provides no temporal repetition, and randomization masks are assigned to coordinates in $\mathcal J_R$. Temporal disturbances can therefore misalign recovered signs with the expected masks. This explains the role of alignment without assuming that one decoded video frame corresponds to one latent time slice.

\noindent\textbf{FrameAverage.}
The sliding-window average blurs pixel values but \textit{does not alter the absolute frame index}: frame $t$ remains at position $t$ after the attack. Because the spatial repetition factors $k_h = 4,\, k_w = 4$ provide strong redundancy, the per-coordinate numerical perturbation is easily corrected by majority voting and zero-thresholding. Hence, Bit Accuracy remains above 0.92.

\noindent\textbf{FrameSwap.}
Swapping adjacent frames with probability $p=0.25$ introduces local ordering changes. The measured video BA is 0.732, 0.699, and 0.710 for VideoShield (ModelScope), VideoShield (LTX-2), and mAVE (LTX-2), respectively. Adjacent swap decisions can interact, and joint generation/inversion need not preserve a one-to-one relation between video frames and latent time slices; a mixture based only on the nominal swap probability is therefore not a quantitative prediction of these rates.

\section{Evaluation Prompt List}
\label{app:prompts}

We curate 250 text prompts from VBench~\citep{huang2024vbench}, evenly distributed across 8 semantic categories. Below we list representative examples from each category. 

\begin{table}[h]
    \centering
    \caption{\textbf{Representative Evaluation Prompts by Category.} Selected examples from each of the 8 VBench categories.}
    \label{tab:prompt_examples}
    \resizebox{\textwidth}{!}{%
    \begin{tabular}{l|p{10cm}}
        \toprule
        \textbf{Category} & \textbf{Example Prompts} \\
        \midrule
        Animal & ``A panda standing on a surfboard in the ocean in sunset''; ``A jellyfish floating through the ocean, with bioluminescent tentacles'' \\
        Architecture & `A 3D model of a 1800s victorian house''; ``Asian garden and medieval castle'' \\
        Food & ``Slow motion cropped closeup of roasted coffee beans falling into an empty bowl''; ``A serving of pumpkin dish in a plate'' \\
        Human & ``Vincent van Gogh is painting in the room''; ``An astronaut feeding ducks on a sunny afternoon, reflection from the water'' \\
        Lifestyle & ``This is how I do makeup in the morning''; ``Home deco with lighted'' \\
        Plant & ``Yellow flowers swing in the wind''; ``An elephant spraying itself with water using its trunk to cool down'' \\
        Scenery & ``Underwater coral reef''; ``Train station platform'' \\
        Vehicles & ``A car stuck in traffic during rush hour''; ``A motorcycle cruising along a coastal highway'' \\
        \bottomrule
    \end{tabular}%
    }
\end{table}

Each prompt is designed to involve both rich visual dynamics and plausible ambient or event-driven audio, ensuring meaningful evaluation of the audio-visual binding capability.

\subsection{Regeneration Attack Protocol}
\label{app:regeneration}
The regeneration row in \cref{tab:joint_robustness} evaluates a content-preserving reconstruction attack, following the diffusion-regeneration setting studied for image watermarks~\citep{an2024waves}. Apply an image diffusion model frame by frame: encode each watermarked frame into its latent space, add Gaussian noise at a low schedule level, reverse-denoise, and decode. Use an empty attack prompt, with no instruction to replace objects or change style. Retain frame order, frame rate, and duration. This is frame-wise purification, not joint video denoising. The noise level controls corruption strength; the number of reverse steps is a separate setting and is not assumed to be one.

Reconstruct the original audio through EnCodec~\citep{defossez2023encodec}, preserving its duration and alignment, then remux the two streams. EnCodec is neural codec reconstruction rather than semantic audio editing. Run authentication on the reconstructed pair using the original generation prompt and session record, with the same frozen detection thresholds as the other attacks. Apply the reconstruction to both authentic and negative pairs when computing authentication rates.

\section{Extended Related Work}
\label{app:related_work}

\subsection{From Single-Modality to Joint Audio-Visual Generation}
\textbf{Diffusion Models and Flow Matching.} Generative AI has been revolutionized by Denoising Diffusion Probabilistic Models~\citep{ho2020denoising} and related image, audio, and video diffusion models~\citep{rombach2022high, saharia2022photorealistic, huang2023make, blattmann2023stable}. Related approaches include Rectified Flow~\citep{liu2022flow}, image-to-image diffusion bridges~\citep{liu2023i2sb}, and consistency models~\citep{song2023consistency}, which provide distinct routes to transport modeling, conditional restoration, and fast generation.

\textbf{Joint Audio-Visual Generation.}
Historically, audio-visual synthesis relied on cascaded pipelines—sequentially generating video-from-audio or audio-from-video. Such pipelines can model cross-modal dependencies through conditioning, but do not jointly denoise the two modalities.
To address these limitations, the field is shifting towards Joint Audio-Visual Generation. State-of-the-art open-source frameworks such as LTX-2~\citep{hacohen2026ltx2} and MOVA~\citep{openmossteam2026mova} process video and audio latents within a joint generative architecture. In these models, a unified denoising loop facilitates bidirectional information exchange via cross-attention, supporting temporal alignment. This shift from independent to joint generation provides the structural basis for our proposed entangled watermarking strategy.

\textbf{Rectified Flow and Invertibility.}
These native bimodal models build upon Rectified Flow~\citep{liu2022flow}, which defines a straight interpolation $\mathbf{z}_t = t\mathbf{x}_1 + (1-t)\mathbf{z}_0$ between data $\mathbf{x}_1$ and noise $\mathbf{z}_0 \sim \mathcal{N}(\mathbf{0}, \mathbf{I})$. As with deterministic DDIM sampling~\citep{song2021ddim}, the associated ODE permits numerical inversion $\mathbf{z}_0 = \text{ODE-Solve}(\mathbf{x}_1, 1{\to}0, v_\theta)$. Learned trajectories need not be exactly straight, and numerical inversion remains approximate. Our solver experiments measure how this approximation affects watermark recovery.

\textbf{Deployment and verification protocols.} Gaussian Shading++~\citep{yang2025gaussianshadingpp} encodes the random seed for pseudorandomization through a second channel based on pseudorandom error-correcting codes. Its soft-decision decoder addresses generation/inversion distortions, and public-key signatures support third-party verification of image watermarks. mAVE instead uses a public record index to retrieve session-specific secrets and verify the relation between two jointly generated modalities. An external requester supplies media and the original prompt or a request reference; the provider performs inversion and verification without releasing its session keys. Thus Gaussian Shading++ emphasizes image-level watermark deployment and public verification, while mAVE establishes a service-mediated test of shared AV generation history. These design choices address different verification workflows and protected objects.

\subsection{Video Watermarking}
As generative capabilities expand, distinguishing synthetic content from authentic footage has become critical. Existing video watermarking schemes broadly fall into two categories: post-processing and in-processing.

\textbf{Post-processing Schemes.}
These methods operate on the decoded video pixels, embedding imperceptible signals via additive perturbations or multiplicative masks~\citep{luo2023dvmark, fernandez2024videoseal}. While effective for legacy copyright protection, they face the \textit{invisibility-robustness trade-off}: increasing robustness against re-encoding often introduces visible artifacts. Furthermore, applying watermarks after generation does not fundamentally verify the provenance of the generative process itself.

\textbf{In-processing Schemes.}
To integrate protection intrinsically, recent works leverage the generative process itself. The paradigm was established by Gaussian Shading~\citep{yang2024gaussian}, which proposed a provable watermarking scheme for image diffusion by coupling the generated output with a fixed initial noise state.
Expanding into the temporal domain, methods like VideoShield~\citep{hu2025videoshield} and the recent VideoMark~\citep{hu2025videomark} generalize these principles to video diffusion, embedding ownership information directly into spatiotemporal latents. The rapid evolution of these algorithms has also led to the development of integrated evaluation libraries such as MarkDiffusion~\citep{pan2025markdiffusion}, which benchmark the performance of various generative watermarking strategies.
However, despite these advancements, current in-processing schemes treat the video latent as an independent variable, disregarding the potential for cross-modal cryptographic binding in bimodal contexts.
\textbf{Trainable video embedding and extraction.} LVMark~\citep{jang2024lvmark} jointly optimizes a latent video decoder and a watermark decoder. Importance-based modulation selects decoder layers for embedding, while its recovery network combines low-frequency 3D-wavelet information with video color features to exploit temporal consistency. VIDSIG~\citep{huang2025vidsig} partially fine-tunes the latent decoder: perturbation-aware suppression freezes perceptually sensitive layers, and a temporal alignment module guides coherent video generation during training. Both integrate learned watermark behavior into video generation. FlowMark~\citep{asnani2026flowmark} instead learns object-mask guidance and region-aware media embedding, with noise-augmented training for robust recovery and temporally coherent watermark placement. Its adaptation takes place in trained watermarking components rather than a joint AV session-binding construction.

\textbf{Recent video-specific designs.} SIGMark~\citep{zhu2026sigmark} combines global frame-wise pseudorandom coding keys with segment group-ordering to support scalable blind extraction and temporal recovery in modern video diffusion models. SKeDA~\citep{yang2026skeda} uses shuffle-key distribution-preserving sampling and differential attention to improve extraction under temporal distortions. Alongside VideoMark, these methods show that training-free embedding and distribution-preserving sampling are not unique to mAVE. Our distinction is the authenticated relation between audio and video from one native joint-generation session, rather than an inability of those methods to watermark modern video generators. These are conceptual comparisons, not additional empirical baselines.

\textbf{Indistinguishability and removal are separate properties.} \citet{lee2025boundary} study removal attacks exploiting boundary information in latent watermark schemes. Their results show why a distributional indistinguishability argument alone does not establish resistance to adversarial removal. In mAVE, public routing positions and the secret authentication payload have different roles: changing the routing signal may disrupt record recovery, while recovering an index does not establish a valid binding. We make no inference from public-index recovery or randomized-region prior preservation to resistance against boundary-guided or oracle-guided removal.

\subsection{Audio Watermarking and The Binding Gap}
Parallel to the visual domain, audio watermarking is essential for safeguarding the auditory component of bimodal content.

\textbf{Deep Generative Audio Watermarking.}
Learned audio watermarking methods include WavMark~\citep{chen2023wavmark}, which uses an invertible neural network for watermark embedding and extraction. A prominent example is AudioSeal~\citep{sanroman2024proactive}, which employs a localized detection mechanism capable of identifying AI-generated segments at the sample level of $1/16,000$s. This localization capability makes it highly effective against splicing and deepfake manipulation in long-form audio.

\textbf{The Binding Gap.}
Despite the maturity of unimodal watermarking, a critical gap remains: \textit{Binding Security}. Current approaches treat audio and video as separate entities, verifying $Video_{wm} \lor Audio_{wm}$ independently. This leaves joint audio-visual generation models vulnerable to Swap Attacks, where an adversary pairs a watermarked video with a mismatched or malicious audio track. mAVE addresses this gap by linking the modalities in their joint initialization space.

\textbf{Cross-Modal Consistency versus Provenance.} SyncNet~\citep{chung2016syncnet} assesses lip--speech synchronization, ImageBind~\citep{girdhar2023imagebind} supplies shared multimodal representations, and CLAP~\citep{laionclap2023} measures audio--text compatibility. These tools provide useful proxy consistency checks, but their scores do not encode generation-session identity. Semantically compatible audio and video can still originate from different requests. We therefore evaluate consistency-based combinations as proxy baselines alongside unimodal watermarking, distinguishing perceptual compatibility from mAVE's same-session authentication objective.

\textbf{Relation to cross-modal media watermarking.} V$^2$A-Mark~\citep{zhang2024v2amark} shares copyright information between video and audio and combines their decoded bits, while \citet{kim2025crossmodal} embed an audio spectrogram into visual-frame representations and recover it for feature-based tamper localization. These are learned media-domain embedding/recovery pipelines. Our architectural advantage concerns native integration with a joint DiT's initialization and inversion, rather than an inability of those methods to post-process DiT outputs. mAVE complements these copyright and integrity protections with generation-session authentication, addressing the risk that separately valid media are misattributed as an original joint output.

\subsection{Fingerprints, Joint AV Signals, and Provenance Infrastructure}
\label{app:provenance_ecosystem}
\textbf{Content-based fingerprinting.} Landmark-based audio identification~\citep{wang2003fingerprint} matches query recordings to a reference database using robust signal descriptors. Video hashes such as TMK+PDQF~\citep{facebook2019tmk} similarly support video similarity matching. Registering both tracks under one database entry can associate a known AV pair. The evidence comes from correspondence to reference content and the matching rule; mAVE instead embeds a keyed relation during joint generation and tests it against a session record. Fingerprinting is therefore useful for reference retrieval, while mAVE supplies a generation-time session signal. Both complement semantic compatibility with reference-based evidence.

\textbf{Joint AV evidence beyond watermarks.} Synchronization and shared-embedding methods quantify perceptual correspondence; AVFF~\citep{oorloff2024avff} learns cross-modal features for classifying manipulated video. Such passive forensic evidence can flag suspicious audio-video combinations without generation-time embedding. mAVE supplies positive evidence of membership in a provider's generation session, including when both modalities are synthetic, complementing consistency and manipulation scores. A provenance application can present these signals together while distinguishing content resemblance, manipulation evidence, and authenticated session association.

\textbf{Cryptographic manifests.} C2PA~\citep{c2pa2026spec} carries signed provenance assertions and cryptographic asset bindings. Its soft-binding mechanism uses fingerprints or watermarks to associate content with manifests, including discovery after embedded metadata is lost. C2PA can describe joint AV assets and their ingredients; mAVE contributes an embedded same-session test.

\textbf{Proposed integration.} A provider could associate each mAVE session record with a signed manifest identifier. Given the original prompt or request context (\cref{sec:detection}), the verifier recovers $I$, checks AV binding, and resolves the issuer-scoped manifest reference. The client validates the manifest signature and applicable asset bindings separately. For transformed renditions, an updated manifest can reference the original as an ingredient and bind the current asset. The mAVE match supplies the generation-session association. Session secrets remain server-side. A standardized soft-binding integration would additionally define the algorithm identifier and resolution interface. This workflow combines media-carried session authentication with signed provenance history.

\subsection{Training and Verification Objectives}
\label{app:trainable_comparison}
\cref{tab:trainable_comparison} separates where a watermark is introduced from what its verifier establishes. Learned video methods develop robust video signatures or recoverable messages; cross-modal hiding schemes additionally recover media content or localize manipulation. mAVE uses fixed generator weights and a server-held session reference to authenticate an audio-video pairing. These objectives are complementary: a copyright message, a tamper map, and a shared-generation-session decision provide different information about the same media.

\begin{table}[ht]
\centering
\small
\setlength{\tabcolsep}{4pt}
\renewcommand{\arraystretch}{1.15}
\caption{\textbf{Architectural and Verification-Objective Comparison.} Training describes watermark-specific optimization, beyond the pretrained media generator.}
\label{tab:trainable_comparison}
\begin{tabular}{@{}>{\raggedright\arraybackslash}p{0.19\linewidth}>{\raggedright\arraybackslash}p{0.39\linewidth}>{\raggedright\arraybackslash}p{0.36\linewidth}@{}}
\toprule
Method & Embedding and training & Verification objective\\
\midrule
LVMark~\citep{jang2024lvmark} & Learned latent-decoder modulation and watermark decoder & Video watermark-message recovery for attribution\\
VIDSIG~\citep{huang2025vidsig} & Partial latent-decoder fine-tuning with temporal alignment & Video signature detection and recovery\\
FlowMark~\citep{asnani2026flowmark} & Trainable mask-guided media embedding & Video message recovery and integrity protection\\
V$^2$A-Mark~\citep{zhang2024v2amark} & Learned video/audio watermarking modules & Copyright recovery and manipulation localization\\
\citet{kim2025crossmodal} & Trainable invertible cross-modal hiding & Audio recovery and audiovisual tamper localization\\
mAVE & Joint latent initialization; fixed generator weights; no watermark-specific training & Same-session AV pair authentication via joint inversion and a server record\\
\bottomrule
\end{tabular}
\end{table}

\FloatBarrier
\section{Session-Binding Controls and Swap Construction}
\label{app:binding_controls}

This section specifies the additional evaluations in \cref{tab:swap_defense,tab:binding_controls}.

\subsection{Shared Session-ID Control}
For each generation session, sample a random tag $s\in\{0,1\}^{N}$ with $N=128$ in the default comparison and store it with the session record for ideal-grid reconstruction. Embed the same $s$ into both modalities using the same latent carriers, repetition budget, inverse sampling, and modality-specific randomization as mAVE. The remaining payload allocation and backbone are held fixed. The tag is session-specific and is not the vendor-wide ownership message or the public record index $I$. The control has no functional dependency of its audio tag on the video digest.

The verifier performs the same prompt-conditioned inversion and session lookup, recovers $\hat s_v,\hat s_a$, and computes $S_{ID}=N^{-1}\sum_j\mathbb{1}[\hat s_v[j]=\hat s_a[j]]$. It accepts when both modality-validity BA gates pass and $S_{ID}>\tau_{ID}$. The agreement threshold is selected on the clean validation split and then frozen. The original aggregate evaluation populates the Shared Session-ID row in \cref{tab:swap_defense}; the separately constructed hard-swap evaluation populates \cref{tab:binding_controls}.

\subsection{Mechanism Ablations}
\textbf{No binding gate.} Reuse the full mAVE watermarked samples and recovered scores. Change only the final verdict to $(\mathrm{BA}_v>\tau_{acc})\land(\mathrm{BA}_a>\tau_{acc})$. This isolates the contribution of the explicit binding-score gate without changing the embedding or generated content.

\textbf{No hash dependency or binding gate.} Generate a new matched sample set. Replace the digest-dependent audio binding bits with pseudorandom bits derived from the audio session sub-key, with no dependence on $B_v$ or its digest. Retain session-key derivation, carriers, payload length, repetition, prompts, seeds, backbone, and solver settings. Verify the recovered modality bits against their corresponding ideal grids using the two BA gates. This tests whether session-specific watermark validity alone accounts for swap rejection. Unlike removing only the gate, it also removes the hash dependency at embedding time. It is a combined ablation; SR-Rand below isolates the binding-string construction while retaining the gate.

\subsection{Matched Server-Reference Controls}
\label{app:server_controls}
The decoded-tag agreement control compares two noisy estimates, whereas mAVE compares recovered binding bits with a server-reconstructed target. The following controls distinguish the reference choice from the construction of the binding string. Both use the same record-retrieval protocol, public-index layout, total payload, repetition factors, backbone, prompt/seed assignments, and joint-inversion settings as mAVE. Tables~\ref{tab:swap_defense} and~\ref{tab:binding_controls} report their measured rates separately from decoded-tag agreement.

\textbf{SR-Rand: random binding with the same gate.} For each session, independently sample and store a uniform $N$-bit string $r$ unrelated to the video grid. Keep the mAVE video grid unchanged. At the same $N$ audio binding positions, replace the selected video-digest bits by $r$, retaining all other audio bits, randomization, and repetition. Regenerate the output pair with this modified embedding. Let $\hat r_a$ be the recovered audio binding string. The verifier uses
\[
 S_{rand}=\frac1N\sum_{j=1}^N\mathbb 1[\hat r_a[j]=r[j]],
\]
and accepts only if both original BA gates and $S_{rand}>\tau_{rand}$ pass, where $\tau_{rand}$ is the frozen calibration threshold for this score. Ideal audio grids are rebuilt using the stored $r$. This replaces the binding target's origin while retaining its length, positions, noise-free reference, and dedicated gate. Unlike the earlier combined ablation, it does not remove the gate.

\textbf{SR-ID: shared tag with noise-free references.} Use the existing Shared Session-ID embedding: a uniform $N$-bit session tag $s$ occupies dedicated slots in each modality, replacing the same number of payload bits rather than adding capacity. Keep the mAVE total per-modality bit/repetition budget. Query the video-selected session record and compare both recovered tags $\hat s_v,\hat s_a$ to that record's stored $s$, defining
\[
 S_{ID,v}=\frac1N\sum_{j=1}^N\mathbb 1[\hat s_v[j]=s[j]],\qquad
 S_{ID,a}=\frac1N\sum_{j=1}^N\mathbb 1[\hat s_a[j]=s[j]].
\]
Retain the two modality-validity BA gates and additionally require $\min(S_{ID,v},S_{ID,a})>\tau_{SR}$, where $\tau_{SR}$ is selected by the common calibration rule. The audio is not checked against its own donor-session record. On an identical Shared Session-ID sample set, this comparison can reuse the recovered tags to isolate the reference change; newly generated samples must not be substituted silently. SR-ID is a stronger session-tag baseline, while SR-Rand is the direct control for mAVE's binding-string construction.

\textbf{Calibration and evaluation.} Select thresholds using the same held-out clean-validation procedure and a common target recall specified before testing, then freeze them for all swap groups. Report the actual test TPR alongside FPR; do not adjust thresholds to force equal test recall. Retain the modality BA calibration and report its pass rate separately. Use identical session-pair identities and candidate pools across methods, invert each spliced output pair anew with the target video's original prompt, and include every attempted query in the denominator. Report binding-score distributions for authentic pairs and each swap group. Where the clean-recall target differs from the original operating point, recompute all methods at the common target rather than transplanting earlier rates.

\subsection{Swap Groups}
\textbf{Cross-prompt random.} For a target video session $A$, randomly choose audio from another session $B$ with a different prompt.

\textbf{Same-prompt, different-session.} Pair video and audio from independent requests sharing the same prompt but using different session secrets and seeds. The four-seed-per-prompt setup supplies candidate pairs; the matched pair $A=B$ is excluded.

\textbf{Similarity-selected.} For each target video, form a fixed pool of audio from other sessions and select the candidate with the highest ImageBind audio--video similarity. Use the same candidate-session identities across methods and do not select by mAVE's secret binding score. Candidate-pool size: 20.

All groups splice the actual output media and run the full detector on the resulting pair; they do not combine watermark bits extracted from the original matched pairs. The video session's original prompt is provided for the joint inversion. The same original authentic session set supplies the TPR comparison. The controls use the shared evaluation setup, with thresholds frozen before the swap tests.

\section{Notation Glossary}
\label{app:notation}

Subscripts $v$ and $a$ denote video and audio. A hat denotes a decoded estimate; a tilde denotes a value recovered by inversion or an attacked sample, as specified locally. Superscript ``ideal'' denotes a noise-free reconstruction from a session record. \cref{tab:notation} collects the method, detector, proof, and evaluation notation. Some appendix formulas use local conventions; these are distinguished explicitly below.

\begingroup
\small
\setlength{\tabcolsep}{5pt}
\renewcommand{\arraystretch}{1.15}
\begin{longtable}{@{}p{0.27\textwidth}p{0.67\textwidth}@{}}
\caption{\textbf{Symbols and Definitions.}}\label{tab:notation}\\
\toprule
Symbol & Meaning \\
\midrule
\endfirsthead
\multicolumn{2}{l}{\textbf{Table \thetable{} (continued)}}\\
\toprule
Symbol & Meaning \\
\midrule
\endhead
\bottomrule
\endfoot
\multicolumn{2}{@{}l}{\textbf{Sessions, keys, and operators}}\\
$A,B$ & Distinct generation sessions, each corresponding to one inference request. \\
$I,\hat I$ & Public session-record index and its decoded estimate; repeated at fixed, unencrypted video positions. \\
$m,K_{priv}$ & Server-held session secret, with $K_{priv}=m$; never embedded. \\
$L_I,L_m$ & Bit lengths of the public index and server secret. \\
$\mathcal J_I,\mathcal J_I(b),r_I$ & Fixed public index coordinates, copies of index bit $b$, and copies per index bit. \\
$\mathcal J_R$ & Complementary randomized coordinates; excludes the public index. \\
$d_v,d_a,\eta_v,\eta_{AV}$ & Video/audio scalar latent counts and public-coordinate fractions within video and joint latents. \\
$a_-,n_-,\alpha$ & Accepted and attempted swapped-pair counts, and tail probability in the independent-trial confidence statement. \\
$n_I,R_I$ & Number of attempted samples and exact full-index recovery rate. \\
$\Pr(\cdot),P(\cdot)$ & Probability of an event; distinct from the prompt string $P$. \\
$P,E_P$ & Original generation prompt and its normalized string $E_P=\operatorname{Norm}(P)$. \\
$\operatorname{Norm}$ & Lowercase conversion and whitespace trimming. \\
$K_{sess}$ & Session key derived from the secret and normalized prompt. \\
$K_{grid},K_{mask},K_{perm}$ & Purpose-separated subkeys for grid construction, per-copy masking, and block permutation. \\
$\operatorname{KDF},\mathrm{label}_x$ & Purpose-separated key derivation and distinct application labels. \\
$\lambda_{key},\operatorname{Prefix}_{\lambda_{key}}$ & Fixed secret-hash prefix length and its truncation operator, independent of binding length $N$. \\
$\operatorname{SHA\text{-}256}$ & Cryptographic hash used for key derivation and video-grid digests. \\
$\operatorname{HMAC}$ & Hash-based message authentication code, used to derive modality-specific pseudorandom base bits. \\
$\parallel,\oplus$ & Bit-string concatenation and bitwise exclusive OR. \\
$\mathrm{seed}_{shared}$ & Shared 64-bit seed extracted from the first eight bytes of the secret hash. \\
$S_{shared}$ & Shared pseudorandom base-bit stream seeded from the secret. \\
$\mathcal K,\mathcal S$ & Masking bit keystream and its ChaCha20 byte-stream representation. \\
\multicolumn{2}{@{}l}{\textbf{Latent spaces, bit grids, and embedding}}\\
$\mathcal Z_v,\mathcal Z_a,\mathcal Z$ & Initial video/audio latent spaces and their joint product space. \\
$\mathcal M$ & Set of initial latent pairs satisfying the session-binding predicate. \\
$\mathbf z_v,\mathbf z_a$ & Video and audio latent tensors; $\mathbf z$ denotes a generic or joint latent. \\
$C,T,H,W$ & Latent tensor sizes along channels, time, height, and width. In temporal-attack formulas, $T$ instead counts video frames. \\
$B_v,B_a,B,L$ & Video/audio bit grids, a generic grid, and grid length. \\
$B_v^{ideal},H_{ideal}$ & Noise-free video grid reconstructed from the session key, and its cryptographic digest used as the binding reference. \\
$h_v$ & Digest $\operatorname{SHA\text{-}256}(B_v)$ inserted into the audio grid at generation time. \\
$\mathbf t,c_0$ & Synchronization template and its designated video channel. Bold $\mathbf t$ is distinct from scalar flow time $t$. \\
$\mathcal I_{inject}$ & Operator inserting template and index bits into designated video-grid slots. \\
$\mathcal I_{time},\mathcal I_{bind}$ & Sets of synchronization-template and audio-binding positions. \\
$\phi(i),N$ & Digest-bit position assigned to binding slot $i$, and number $N=|\mathcal I_{bind}|$ of independent comparisons. \\
$\operatorname{Base}_a$ & Audio sub-key-derived base grid at positions not reserved for template or binding bits. \\
$B^{diff},M_{rand}$ & Repeated bit tensor and sign-bit map: XOR-masked on $\mathcal J_R$, unmasked on $\mathcal J_I$. \\
$k_c,k_t,k_h,k_w$ & Repetition factors along the four latent axes; $k_{all}=k_ck_tk_hk_w$ is the copy count per bit. \\
$\Pi$ & Randomized inverse-CDF map from bit grids to Gaussian latents. \\
$\Phi,\Phi^{-1},\mathrm{ppf}$ & Standard normal CDF and its inverse (quantile function). \\
$u_i,u_i^\delta,z_i$ & Auxiliary uniform variate, its half-interval clamp, and the resulting latent coordinate; Appendix E specifies the numerical rule. \\
$\mathcal I_w,y_i,\pi$ & Watermark-coordinate set, assigned binary value, and coordinate mapping induced by the secret block permutation in Appendix E. In Appendix F, $\pi(j,l)$ locates copy $l$ of payload bit $j$. \\
$\delta$ (sampling) & Uniform-input clipping margin, set to $0.001$ in Appendix E; distinct from inversion cutoff $\delta_t$. \\
\multicolumn{2}{@{}l}{\textbf{Inversion, authentication, and evaluation}}\\
$\mathbf x_v^{(A)},\mathbf x_a^{(B)}$ & Video from session $A$ and audio from session $B$; their splice is $\tilde{\mathbf x}$. \\
$\mathbf x_{query}$ & Queried pair in the generator's joint data-latent representation, used as the inversion endpoint. \\
$t,\mathbf z_t$ & Flow time and the joint latent state at that time. \\
$v_\theta,\theta,c$ & Flow velocity field, generator parameters, and condition from the original prompt. \\
$\delta_t,\tilde{\mathbf z}_0$ & Minimum inversion timestep and recovered approximation to the initial latent. \\
$\mathcal D,\operatorname{Verify}$ & Final authentication detector and binding predicate; the detector combines modality-validity and binding gates. \\
$S_{bind},\tau_{bind}$ & Fraction of decoded audio-binding bits matching the ideal digest, and its acceptance threshold. \\
$\mathrm{BA}_v,\mathrm{BA}_a,\tau_{acc}$ & Video/audio bit accuracies against ideal grids, and the modality-validity threshold. \\
$H_0,H_1$ & Same-session pair and mismatched-session pair hypotheses. \\
$\mathrm{TPR},\mathrm{FNR}$ & True-positive and false-negative rates on authentic pairs. \\
$\mathrm{TNR},\mathrm{FPR}$ & True-negative and false-positive rates on swapped pairs in the authentication protocol. \\
$\mathrm{Acc}$ & Balanced accuracy $(\mathrm{TPR}+\mathrm{TNR})/2$ in Table 3. \\
$n_s,\mathcal C,M$ & Session-pool size, ordered cross-session candidate family, and its size $M=n_s(n_s-1)$. \\
$Q_{ab},Q_0,Q_\delta$ & Ordered swap query, fixed baseline swap, and locally perturbed donor-audio query. \\
$\kappa,r_\tau,d_H$ & Minimum matching-bit count, maximum allowed mismatch count, and Hamming distance. \\
$\varepsilon_N(b)$ & Uniform Hamming-ball probability at radius $\min(N,r_\tau+b)$; $b=0$ gives the exact baseline binding tail. \\
$b,\beta,\beta_g$ & Fixed decoded-bit change budget, per-pair stability-failure probability bound, and pool-wide failure bound. \\
$U,U_0,U_\delta,h$ & Recovered binding vector, baseline/perturbed versions, and selected target digest bits. \\
$\mathcal G_b,\mathcal G_b^{pool}$ & Events that all allowed perturbations preserve routing and obey the bit-change budget, for one pair or a pool. \\
$K,V,\mathcal D_K,q$ & Fixed verification secrets, adversary side information, stateless binary verifier, and verification-query budget. \\
$\mathcal A_t,\mathrm{Forge}$ & Adversary next-query rule and event that an admissible mismatched candidate is accepted. \\
$\varepsilon_{single},Q_t^\perp$ & Best allowed no-verification-feedback success bound and candidate $t$ on the all-reject path. \\
$\Delta_{\epsilon_{atk}},\epsilon_{atk}$ & Allowed perturbation set and total input budget relative to the original donor audio. \\
$\mathcal F,L_{\mathcal F},\eta$ & Complete numerical inversion map, its certified local norm bound, and resulting latent perturbation bound. \\
$G_i,s_j,\gamma_i,e_i,b_{cert}$ & Binding-bit copy group, derandomized sign, vote margin, potentially changing copy count, and certified logical-bit change count. \\
$\nu$ & Adversary-success distinguishing gap supplied by a real-to-ideal reduction, when available. \\
$s,\hat s_v,\hat s_a$ & Shared Session-ID tag and its decoded video/audio estimates (Appendix J). \\
$r,\hat r_a,S_{rand},\tau_{rand}$ & Stored random binding target, recovered audio target, matching score, and threshold in SR-Rand (Appendix J). \\

$S_{ID,v},S_{ID,a},\tau_{SR}$ & Video/audio matches to the stored tag and common score threshold in SR-ID. \\
$S_{ID},\tau_{ID}$ & Cross-modal tag agreement and its frozen threshold in the Shared Session-ID control. \\
T2AV, I2AV, TI2AV & Text-to-AV, image-to-AV, and text-and-image-to-AV generation. \\
\multicolumn{2}{@{}l}{\textbf{Theory and appendix-local conventions}}\\

$\mathcal A,\mathcal B$ & Efficient distinguishing algorithms in the randomization reduction. \\

$I$ in $\mathcal N(0,I)$ & Identity covariance matrix, not the public session index. \\
$f,\mathbb 1[\cdot]$ & Bounded output statistic in Appendix A.1, and indicator function. \\

$l,M,U,Z,s_l$ & Bits per coordinate, sign/bin vector, independent uniform variates, sampled vector, and bin map in Appendix A.1; $M$ there is local, not the swap-pool size. \\
$\Gamma,M_R,R,Y^s,Y^{ref}$ & Public context, real randomized authentication signs, uniform replacement signs, and real/matched numerical pipeline outputs. \\
$\epsilon_{rand},\epsilon_{sep},\epsilon_{stream}$ & Joint randomization advantage and its purpose-separation/stream-generator hybrid contributions. \\
$\mathcal S_{num},\mathcal P_{num},\omega$ & Numerical sampler, sampler optionally composed with the generator, and independent implementation coins. \\
$\varphi_d,\gamma_d,\operatorname{TV},D_{KL},H$ & Standard Gaussian density/law, total variation, relative entropy, and entropy in natural-log units. \\
$D,d_I,d_R,C(I)$ & Full scalar latent count, public/randomized region counts, and repeated index code. \\
$\ell_\delta,a_\delta,H_\delta,T_\delta$ & Magnitude floor/cap, magnitude clamp, and signed clamp. \\
$\chi_\delta,v_\delta,W_2$ & Clamp mean-square displacement, clamped second moment, and quadratic Wasserstein distance. \\
$P_\delta^{ref},P_{num}^{ref},\zeta_{num}$ & Continuous clamped reference, numerical matched reference, and their normalized coupling discrepancy. \\
$F,L_F,X_G$ & Expected output statistic as a function of latent input, its normalized Lipschitz constant, and an output from standard Gaussian initialization. \\
$N,t_i,\epsilon$ (Appendix B) & Solver step count, discrete times, and final time $\epsilon=\delta_t$. This local $N$ is not the binding length. \\
$\alpha_t,\epsilon_\theta$ & DDIM signal-schedule coefficient and predicted noise in Appendix B. \\
$\mu_j,\mathrm{Block}_j$ & Median of repeated latent values for bit $j$, and its repetition block (Appendix C). \\
$m_j,\hat m_j,\tilde b_{j,l}$ & Generic payload bit, decoded bit, and recovered copy in Appendix F; these $m_j$ do not denote the server secret. \\
$d_i,c_i$ & Normalized DINO and CLIP frame features used for subject/background consistency; $c_i$ is not prompt condition $c$. \\
$f_t,f'_t$ & Original and attacked video frames in Appendix G. \\
$n,p$ & Temporal averaging radius and adjacent-frame swap probability. \\
\end{longtable}
\endgroup

\end{document}